\documentclass[11pt]{article}
\usepackage{geometry}
\usepackage{amsmath, amsfonts}
\usepackage{amsthm}

\newtheorem{fact}{Fact}

\newtheorem{corollary}{Corollary}
\newtheorem{theorem}{Theorem}
\newtheorem{lemma}{Lemma}

\newtheorem{definition}{Definition}
\newtheorem{proposition}{Proposition}

\usepackage[russian,greek,english]{babel}
\usepackage[utf8x]{inputenc}
\usepackage{amsmath,ifthen,color,eurosym}

\usepackage{wrapfig,hyperref,cite}
\usepackage{latexsym,amsmath,amssymb,url}
\usepackage{fancyhdr}
\usepackage{graphicx}
\usepackage{stmaryrd,wasysym}
\usepackage{colortbl,datetime}
\usepackage{boxedminipage}
\usepackage{slashbox,xcolor,dsfont}

\usepackage[lined,boxed,commentsnumbered]{algorithm2e}

\newcommand{\remove}[1]{}

\newcommand{\tw}{\mbox{\bf tw}}

\renewcommand{\int}{{\bf int}}

\renewcommand{\L}[1]{\Lambda#1}

 \usepackage{todonotes}

\renewcommand{\leq}{\leqslant}
\renewcommand{\geq}{\geqslant}
\renewcommand{\epsilon}{\varepsilon}
\newtheorem{rgl}  {Rule}

\newcommand{\mc}{\mathcal}
\renewcommand{\G}{\ensuremath{\mc{G}}\xspace}
\newcommand{\F}{\ensuremath{\mc{F}_t}\xspace}
\newcommand{\B}{\ensuremath{\mc{B}_t}\xspace}

\newcommand{\Epi}[2]{\ensuremath{\mc{#1E}_{#2}}\xspace}
\newcommand{\E}{\Epi{}{}}
\newcommand{\Enospace}{\Epi{}{}}
\renewcommand{\C}{\ensuremath{\mc{C}^{\E}}\xspace}

\renewcommand{\f}{\ensuremath{f^{\E}}\xspace}
\renewcommand{\L}{\ensuremath{L^{\E}}\xspace}

\newcommand{\Ebar}{{\Epi{\bar}{g}}\xspace}

\newcommand{\fbar}{\ensuremath{\bar f^{\E}_{g}}\xspace}

\newcommand{\eq}[1]{\ensuremath{\sim_{#1,t}}\xspace}
\newcommand{\eqG}[1]{\ensuremath{\sim_{#1,\G,t}}\xspace}
\newcommand{\eqs}[1]{\ensuremath{\eq{#1}^*}\xspace}
\newcommand{\eqsG}[1]{\ensuremath{\eqG{#1}^*}\xspace}
\newcommand{\neqG}[1]{\ensuremath{\nsim_{#1,\G,t}}\xspace}
\newcommand{\D}[1]{\ensuremath{\Delta_{#1,t}}\xspace}

\newcommand{\YES}{\textsc{Yes}\xspace}
\newcommand{\NO}{\textsc{No}\xspace}
\newcommand{\FPack}{\textsc{$\cal F$-Packing}\xspace}
\newcommand{\rFPack}{\textsc{$\ell$-$\cal F$-Packing}\xspace}
\renewcommand{\tw}{{\mathbf{tw}}}

\usepackage[lined,boxed,commentsnumbered]{algorithm2e}

\usepackage{graphicx}

\begin{document}

\title{Explicit linear kernels for packing problems\thanks{{Emails:  Valentin Garnero: {\sf valentin.garnero@inria.fr},
Christophe Paul: {\sf paul@lirmm.fr},
Ignasi Sau: {\sf sau@lirmm.fr},
Dimitrios M. Thilikos: {\sf sedthilk@thilikos.info}}.}}

\author{Valentin Garnero\thanks{AlGCo project-team, CNRS, LIRMM, Universit\'e de Montpellier, Montpellier, France.}\and Christophe Paul$^{\dagger}$\and Ignasi Sau$^{\dagger}$\and Dimitrios  M. Thilikos$^{\dagger}$\thanks{Department of Mathematics, National and Kapodistrian University of Athens, Athens, Greece.}}

\date{\empty}
 \maketitle

\begin{abstract}
\noindent During the last years, several algorithmic meta-theorems have appeared (Bodlaender \emph{et al}. [FOCS 2009], Fomin \emph{et al}. [SODA 2010], Kim \emph{et al}. [ICALP 2013]) guaranteeing the {\sl existence} of linear kernels on sparse graphs for problems satisfying some generic conditions. The drawback of such general results is that it is usually not clear how to derive from them {\sl constructive} kernels with reasonably low {\sl explicit} constants. To fill this gap, we recently presented [STACS 2014] a framework to obtain explicit linear kernels for some families of problems whose solutions can be certified by a subset of {\sl vertices}. In this article we enhance our framework to deal with {\sl packing} problems, that is, problems whose solutions can be certified by collections of {\sl subgraphs} of the input graph satisfying certain properties. ${\mathcal F}$-\textsc{Packing} is a typical example: for a family ${\mathcal F}$ of connected graphs  that we assume to contain at least one planar graph, the task is to decide whether a graph $G$ contains $k$ vertex-disjoint subgraphs such that each of them contains a graph in ${\mathcal F}$ as a minor. We provide explicit linear kernels on sparse graphs for the following two orthogonal generalizations of ${\mathcal F}$-\textsc{Packing}: for an integer $\ell \geq 1$, one aims at finding either minor-models that are pairwise at distance at least $\ell$ in $G$ (\textsc{$\ell$-$\mc{F}$-Packing}), or such that each vertex in $G$ belongs to at most $\ell$ minors-models (\textsc{$\mc{F}$-Packing with $\ell$-Membership}). Finally, we also provide linear kernels for the versions of these problems where one wants to pack {\sl subgraphs} instead of minors.

\medskip
\end{abstract}

\setcounter{footnote}{0}
\section{Introduction}
\label{sec:intro}


\textbf{Motivation.} A fundamental notion in parameterized complexity (see~\cite{CyganFKLMPPS15} for a recent textbook) is that of \emph{kernelization}, which asks for the existence of polynomial-time preprocessing algorithms producing equivalent instances whose size depends exclusively  on the parameter $k$. Finding kernels of size polynomial or linear in $k$ (called \emph{linear kernels}) is one of the major goals of this area. A pioneering work in this direction was the linear kernel of  Alber \emph{et al}.~\cite{AFN04} for \textsc{Dominating Set} on planar graphs, generalized by Guo and Niedermeier~\cite{GuNi07} to a family of problems on planar graphs. Several algorithmic meta-theorems on kernelization have appeared in the last years, starting with the result of Bodlaender \emph{et al}.~\cite{BFL+09} on graphs of bounded genus. It was followed-up by similar results on larger sparse graph classes, such as graphs excluding a minor~\cite{FLST10} or a topological minor~\cite{KLP+12}.


The above results guarantee the {\sl existence} of linear kernels on sparse graph classes for problems satisfying some generic conditions, but it is hard to derive from them {\sl constructive} kernels with {\sl explicit} constants. We recently made in~\cite{KviaDP} a significant step toward a fully constructive meta-kernelization theory on sparse graphs with explicit constants. In a nutshell, the main idea is to substitute the algorithmic power of CMSO logic that was used in~\cite{BFL+09,FLST10,KLP+12} with that of dynamic programming (DP for short) on graphs of bounded decomposability (i.e., bounded treewidth). We refer the reader to the introduction of~\cite{KviaDP} for more details. Our approach provides a DP framework able to construct linear kernels for families of problems on sparse graphs whose solutions can be certified by a subset of {\sl vertices} of the input graph,  such as $r$-\textsc{Dominating Set} or \textsc{Planar-$\mathcal{F}$-Deletion}.

%


\vspace{.2cm}
\noindent\textbf{Our contribution.} In this article we make one more step in the direction of a fully constructive meta-kernelization theory on sparse graphs, by enhancing the existing framework~\cite{KviaDP} in order to deal with {\sl packing} problems. These are problems whose solutions can be certified by collections of {\sl subgraphs} of the input graph satisfying certain properties. We call these problems \emph{packing-certifiable}, as opposed to \emph{vertex-certifiable} ones. For instance, deciding whether a graph $G$ contains at least $k$ vertex-disjoint cycles is a typical packing-certifiable problem. This problem, called \textsc{Cycle Packing}, is {{\sf FPT}} as it is minor-closed, but it is unlikely to admit polynomial kernels  on general graphs~\cite{BodlaenderTY11}.

As an illustrative example, for a family of connected graphs ${\mathcal F}$ containing at least one planar graph, we provide a linear kernel on sparse graphs for the \textsc{${\mathcal F}$-Packing} problem\footnote{We would like to clarify here that in our original conference submission of~\cite{KviaDP} we claimed, among other results,  a linear kernel for \textsc{${\mathcal F}$-Packing} on sparse graphs. Unfortunately, while preparing the camera-ready version, we realized that there was a bug in one of the proofs and we had to remove that result from the paper. It turned out that for fixing that bug, several new ideas and a generalization of the original framework seemed to be necessary; this was  the starting point of the results presented in the current article.}: decide whether a graph $G$ contains at least $k$ vertex-disjoint subgraphs such that each of them contains a graph in ${\mathcal F}$ as a minor, parameterized by $k$. We provide linear kernels as well for the following two {\sl orthogonal generalizations} of ${\mathcal F}$-\textsc{Packing}: for an integer $\ell \geq 1$, one aims at finding either minor-models that are pairwise at distance at least $\ell$ in $G$ (\textsc{$\ell$-$\mc{F}$-Packing}), or such that each vertex in $G$ belongs to at most $\ell$ minors-models (\textsc{$\mc{F}$-Packing with $\ell$-Membership}). While only the {\sl existence} of linear kernels for \textsc{${\mathcal F}$-Packing} was known~\cite{BFL+09}, to the best of our knowledge no kernels were known for \textsc{$\ell$-$\mc{F}$-Packing} and \textsc{$\mc{F}$-Packing with $\ell$-Membership}, except for \textsc{$\ell$-$\mc{F}$-Packing} when $\mathcal{F}$ consists only of a triangle and the maximum degree is also considered as a parameter~\cite{AtminasKR14}. We would like to note that the kernels for \textsc{${\mathcal F}$-Packing} and for \textsc{$\mc{F}$-Packing with $\ell$-Membership} apply to minor-free graphs, while those for \textsc{$\ell$-$\mc{F}$-Packing} for $\ell \geq 2$ apply to the smaller class of apex-minor-free graphs.

We also provide linear kernels for the versions of the above problems where one wants to pack {\sl subgraphs} instead of minors (as one could expect, the kernels for subgraphs are considerably simpler than those for minors). We call the respective problems \textsc{$\ell$-${\cal F}$-Subgraph-Packing} and \textsc{$\mc{F}$-Subgraph-Packing with $\ell$-Membership}. While the first problem can be seen as a broad generalization of $\ell$-\textsc{Scattered Set} (see for instance~\cite{BFL+09,KviaDP}), the second one was recently defined by
Fernau~\emph{et al}.~\cite{FernauLR15}, motivated by the problem of discovering overlapping communities (see also~\cite{RomeroL14-WALCOM,RomeroL14-CSR} for related problems about detecting overlapping communities): the parameter $\ell$ bounds the number of communities that a member of a network can belong to. More precisely, the goal is to find in a graph $G$ at least $k$ subgraphs isomorphic to a member of $\mathcal{F}$ such that every vertex in $V(G)$ belongs to at most $\ell$  subgraphs. This
type of overlap was also studied by Fellows \emph{et al}.~\cite{FellowsGKNU11} in the context of graph editing. Fernau~\emph{et al}.~\cite{FernauLR15} proved, in particular, that the \textsc{$\mathcal{F}$-Subgraph-Packing with $\ell$-Membership} problem is {\sc{NP}}-hard for all values of $\ell \geq 1$ when ${\mathcal F}= \{F\}$ and $F$ is an arbitrary connected graph with at least three vertices, but polynomial-time solvable for smaller graphs. Note that \textsc{$\mathcal{F}$-Subgraph-Packing with $\ell$-Membership} generalizes the \textsc{$\mathcal{F}$-Subgraph-Packing} problem, which consists in finding in a graph $G$ at least $k$ vertex-disjoint subgraphs isomorphic to a member of $\mathcal{F}$.  The smallest kernel for the \textsc{$\mathcal{F}$-Subgraph-Packing} problem~\cite{Moser09} has size $O(k^{r-1})$, where $\mathcal{F} = \{F\}$ and $F$ is an arbitrary graph on $r$ vertices. A list of references of kernels for particular cases of the family $\mathcal{F}$ can be found in~\cite{FernauLR15}.
Concerning the kernelization of \textsc{$\mathcal{F}$-Subgraph-Packing with $\ell$-Membership}, Fernau~\emph{et al}.~\cite{FernauLR15}
provided a  kernel on general graphs with $O((r+1)^r k^r)$ vertices, where $r$ is the maximum number of vertices of a graph in $\mathcal{F}$. In this article we improve this result on graphs excluding a fixed graph as a minor, by providing a linear  kernel for \textsc{$\mathcal{F}$-Subgraph-Packing with $\ell$-Membership} when $\mathcal{F}$ is any family of (not necessarily planar) connected graphs.

\vspace{.2cm}
\noindent\textbf{Our techniques: vertex-certifiable vs. packing-certifiable problems}. It appears that packing-certifiable problems are intrinsically more involved than vertex-certifiable ones. This fact is well-known when speaking about {{\sf FPT}}-algorithms on graphs of bounded tree\-width~\cite{CyganNPPRW11,LokshtanovMS11}, but we need to be more precise with what we mean by being ``more involved'' in our setting of obtaining kernels via DP on a tree decomposition of the input graph. Loosely speaking, the framework that we presented in~\cite{KviaDP} and that we need to redefine and extend here, can be summarized as follows. First of all, we
propose a general definition of a problem \emph{encoding} for the tables of DP when solving
parameterized problems on graphs of bounded treewidth. Under this setting, we
provide three general conditions guaranteeing that such an encoding can yield a so-called \emph{protrusion replacer}, which in short is a procedure that replaces large ``protrusions'' (i.e., subgraphs with small treewidth and small boundary) with ``equivalent'' subgraphs of constant size.  Let us be more concrete on these three conditions that such an encoding $\E$ needs to satisfy in order to obtain an explicit linear kernel for a parameterized problem $\Pi$. 

The first natural condition is that on a graph $G$ without boundary, the optimal size of the objects satisfying the constraints imposed by $\E$ coincides with the optimal size of solutions of $\Pi$ in $G$; in that case we say that $\E$ is a \emph{$\Pi$-encoder}. On the other hand, we need that when performing DP using the encoding $\E$, we can use tables such that the maximum difference among all the values that need to be stored is bounded by a function $g$ of the treewidth; in that case we say that $\E$ is \emph{$g$-confined}. Finally, the third condition requires that $\E$ is ``suitable'' for performing DP, in the sense that the tables at a given node of a tree decomposition can be computed using only the information stored in the tables of its children (as it is the case of practically all natural DP algorithms); in that case we say that $\E$ is \emph{DP-friendly}. These two latter properties exhibit some fundamental differences when dealing with vertex-certifiable or packing-certifiable problems.

 Indeed, as discussed in more detail in Section~\ref{sec:generic}, with an encoding $\E$ we associate a function $f^\E$ that corresponds, roughly speaking, to the maximum size of a partial solution that satisfies the constraints
defined by $\E$. In order for an encoder to be $g$-confined for some function $g(t)$ of the treewidth $t$, for some vertex-certifiable problems such as $r$-\textsc{Scattered Set} (see~\cite{KviaDP}) we need to ``force'' the confinement artificially,
in the sense that we directly discard the entries in the tables whose associated values differ by more than $g(t)$ from the
maximum (or minimum) ones. Fortunately, we can prove that an encoder with this modified function is still DP-friendly.
However, this is not the case for packing-certifiable problems such as \textsc{$\mathcal{F}$-Packing}. Intuitively, the
difference lies on the fact that in a packing-certifiable problem, a solution of size $k$ can contain arbitrarily many
vertices (for instance, if one wants to find $k$ disjoint cycles in an $n$-vertex graph with girth $\Omega(\log n)$)
and so it can as well contain arbitrarily many vertices from any subgraph corresponding to a rooted subtree of a tree
decomposition of the input graph $G$. This possibility prevents us from being able to prove that an encoder is DP-friendly
while still being $g$-confined for some function $g$, as in order to fill in the entries of the tables at a given node, one may need to retrieve information
from the tables of other nodes different from its children.  To circumvent this problem, we introduce another criterion to
discard the entries in the tables of an encoder: we recursively discard the entries of the tables whose associated partial
solutions {\sl induce} partial solutions at some lower node of the rooted tree decomposition that need to be discarded. That
is, if an entry of the table needs to be discarded at some node of a tree decomposition, we {\sl propagate} this information
to all the other nodes.

\vspace{.20cm}
\noindent \textbf{Organization of the paper.}
Some basic preliminaries can be found in Section~\ref{sec:prelim}, including graph minors, parameterized problems, (rooted) tree decompositions, boundaried graphs, the canonical equivalence relation $\equiv_{\Pi,t}$ for a problem $\Pi$ and an integer $t$, FII, protrusions, and protrusion decompositions. The reader not familiar with the background used in previous work on this topic may see~\cite{BFL+09,FLST10,KLP+12,KviaDP}. In Section~\ref{sec:generic} we introduce the basic definitions of our framework and present an explicit protrusion replacer for packing-certifiable problems. Since many definitions and proofs in this section are quite similar to the ones we presented in~\cite{KviaDP}, for better readability we moved the proofs of the results marked with `$[\star]$' to Appendix~\ref{ap:framework}.
Before moving to the details of each particular problem, in Section~\ref{sec:applications} we summarize the main ingredients that we use in our applications. The next sections are devoted to showing how to apply our methodology to various families of problems. More precisely, we start in Section~\ref{sec: FPack} with the linear kernel for \textsc{Connected-Planar-$\mc{F}$-Packing}.
This problem is illustrative, as it contains most of the technical ingredients of our approach, and will be generalized later in the two orthogonal directions mentioned above. Namely, in Section~\ref{sec: rFPack} we deal with the variant in which the minor-models are pairwise at distance at least $\ell$, and in Section~\ref{sec: FMemb} with the version in which each vertex can belong to at most $\ell$ minor-models. In Section~\ref{sec: FSub} we adapt the machinery developed for packing minors to packing subgraphs, considering both variants of the problem. For the sake of completeness, each of the considered problems will be redefined in the corresponding section. Finally, Section~\ref{sec:conclusions} concludes the article.

%

\section{Preliminaries}
\label{sec:prelim}


In  our article graphs are undirected, simple, and without loops. We use standard graph-theoretic notation; see for instance~\cite{Die05}.
We denote by $d_G(v,w)$ the distance in $G$ between two vertices $v$ and $w$ and by $d_G(W_1,W_2) = \min\{ d_G(w_1,w_2) : w_1\in W_1, w_2\in W_2 \}$ the distance between two sets of vertices $W_1$ and $W_2$ of $G$. Given $S \subseteq V(G)$, we denote by $N(S)$ the set of vertices in $V(G) \setminus S$ having at least one neighbor in $S$.



\begin{definition} A parameterized graph problem ${\Pi}$ is called \emph{packing-certifiable}
if there exists a language $L^{\Pi}$ (called {\em certifying language for $\Pi$}) defined on pairs $(G,{\cal S})$, where $G$ is a graph and ${\cal S}$ is a collection of subgraphs of $G$, such that $(G,k)$ is a \YES-instance of $\Pi$ if and only if there exists a collection ${\cal S}$ of subgraphs of $G$ with $|{\cal S}| \geq k$ such that $(G,{\cal S}) \in L^{\Pi}$.
\end{definition}

In the above definition, for the sake of generality we do not require the subgraphs in the collection ${\cal S}$ to be pairwise distinct. Also, note that the subclass of packing-certifiable problems where each subgraph in ${\cal S}$ is restricted to consist of a single vertex corresponds to the class of vertex-certifiable problems defined in~\cite{KviaDP}.

For a class of graphs \G, we denote by $\Pi_\G$  the problem $\Pi$  where the instances are restricted to contain graphs belonging to \G. With a packing-certifiable problem we can associate in a natural way an optimization function as follows.


\begin{definition}\label{defi:optimization-function}
Given a packing-certifiable parameterized problem $\Pi$,  the \emph{maximization function} $f^{\Pi}: \Gamma^* \rightarrow \mathbb{N}\cup\{-\infty\}$ is defined as
\begin{align}
f^{\Pi}(G)=\
\left\{\begin{array}{lll}
  & \max\{|{\cal S}| : (G,{\cal S}) \in L^{\Pi}\} & \mbox{, if there exists such an ${\cal S}$ and}\\
  & - \infty                     & \mbox{, otherwise}.
\end{array}\right.
\end{align}
\end{definition}



\begin{definition}  \label{defi:boundaried}
A \emph{boundaried graph} is a graph $G$ with a set $B \subseteq V (G)$ of distinguished vertices and an injective labeling $\lambda_G: B \to \mathbb{N}$. The set $B$ is called the \emph{boundary} of $G$ and it is denoted by $\partial(G)$. The set of labels is denoted by $\Lambda(G) = \{\lambda_G(v) : v \in \partial(G) \}$. We say that a boundaried graph is a \emph{$t$-boundaried graph} if $\Lambda(G) \subseteq \{1, \ldots ,t\}$.
\end{definition}

We denote by \B the set of all $t$-boundaried graphs.

\begin{definition} \label{defi:gluing}
Let $G_1$ and $G_2$ be two boundaried graphs. We denote by $G_1 \oplus G_2$ the graph obtained from $G$ by taking the disjoint union of $G_1$ and $G_2$ and identifying  vertices with the same label in the boundaries of $G_1$ and $G_2$. In $G_1 \oplus G_2$ there is an edge between two labeled vertices if there is an edge between them in $G_1$ or in $G_2$.
\end{definition}



Given $G = G_1 \oplus G_2$ and $G_2'$, we say that $G' = G_1 \oplus G_2'$ is the graph obtained from $G$ by \emph{replacing} $G_2$ with $G_2'$. The following notion was introduced  by Bodlaender \emph{el al}.~\cite{BFL+09}.



\begin{definition}  \label{defi:cano}
Let $\Pi$ be a parameterized problem and let $t \in \mathbb{N}$.
Given $G_1,G_2 \in \B$, we say that $G_1 \equiv_{\Pi} G_2$ if $\Lambda(G_1) = \Lambda(G_2)$ and there exists a transposition constant $\D\Pi(G_1,G_2) \in \mathbb{Z}$ such that for every $H \in \B$ and every $k \in \mathbb{Z}$, it holds that $(G_1 \oplus H, k) \in \Pi$ if and only if $(G_2 \oplus H, k+\D\Pi(G_1,G_2)) \in \Pi$.
\end{definition}


\begin{definition}
A \emph{tree decomposition} of a graph $G$ is a couple $(T,\mc{X} = \{ B_x : x \in V(T) \})$, where $T$ is a tree and such that $\bigcup_{x \in V(T)} B_x = V(G)$, for every edge $\{u,v\} \in E (G)$ there exists $x \in V(T)$ such that $u,v \in B_x$, and
 for every vertex $u \in V(G)$ the set of nodes $\{ x\in V(T) : u \in B_x \}$ induce a subtree of $T$.
The vertices of $T$ are referred to as \emph{nodes} and the sets $B_x$ are called bags.

A \emph{rooted tree decomposition} $(T,\mc{X},r)$ is a tree decomposition with a distinguished node $r$ selected as the \emph{root}. A \emph{nice tree decomposition} $(T,\mc{X},r)$ (see~\cite{Klo94}) is a rooted tree decomposition where $T$ is binary and for each node $x$ with two children $y,z$ it holds $B_x =B_y =B_z$ and for each node $x$ with one child $y$ it holds $B_x =B_y \cup \{u\}$ or $B_x =B_y \setminus \{ u\}$ for some $u \in V(G)$. The \emph{width} of a tree decomposition is the size of a largest bag minus one. The \emph{treewidth} of a graph, denoted by $\tw(G)$, is the smallest width of a tree decomposition of $G$. A \emph{treewidth-modulator} of a graph $G$ is a set $X \subseteq V(G)$ such that $\tw(G-X) \leq t$, for some fixed constant $t$.
\end{definition}

Given a bag $B$ (resp. a node $x$) of a rooted tree decomposition $T$, we denote by $G_B$ (resp. $G_x$), the subgraph induced by the vertices appearing in the subtree of $T$ rooted at the node corresponding to $B$ (resp. the node $x$). We denote by \F the set of all $t$-boundaried graphs that have a rooted tree decomposition of width $t-1$ with all boundary vertices contained in the root-bag. Obviously $\F \subseteq \B$. (Note that graphs can be viewed as 0-boundaried graphs, hence we use a same alphabet $\Gamma$ for describing graphs and boundaried graphs.) 



\begin{definition} \label{defi:prot}
Let $t,\alpha$ be positive integers. A \emph{$t$-protrusion} $Y$ of a graph $G$ is an induced subgraph of $G$ with $|\partial(Y)| \leq t$ and $ \tw(Y) \leq t-1$, where $\partial(Y)$
is the set of vertices of $Y$ having neighbors in $V(G) \setminus V(Y)$.
An \emph{$(\alpha,t)$-protrusion decomposition} of a graph $G$ is a partition
    ${\cal P}=Y_{0}\uplus Y_{1}\uplus \cdots \uplus Y_{\ell}$ of $V(G)$ such
    that for every $1\leqslant i\leqslant \ell$, $N(Y_{i})\subseteq Y_{0}$,
    $\max\{\ell, |Y_{0}|\}\leqslant \alpha$, and
     for every $1\leqslant i\leqslant \ell$, $Y_i\cup N(Y_i)$  is a $t$-protrusion of $G$. When $(G,k)$ is the input of a parameterized problem with parameter $k$, we say that an
$(\alpha,t)$-protrusion decomposition of $G$ is \emph{linear} whenever $\alpha =O(k)$.
\end{definition}

%


We say that a rooted tree decomposition of a protrusion $G$ (resp. a boundaried graph $G$) is \emph{boundaried} if the boundary $\partial(G)$ is contained in the root bag. In the following we always consider boundaried nice tree decompositions of width $t-1$,  which can be computed in polynomial time for fixed $t$~\cite{Klo94, Bod96}.



\section{A framework to replace protrusions for packing problems}
\label{sec:generic}

In this section we restate and in many cases modify the definitions given in~\cite{KviaDP} in order to deal with packing-certifiable problems; we will point out the differences. As announced in the introduction, missing proofs can be found in Appendix~\ref{ap:framework}.


\renewcommand{\f}{\ensuremath{f^{\E}}\xspace}
\renewcommand{\L}{\ensuremath{L^{\E}}\xspace}



\smallskip\smallskip
\noindent \textbf{Encoders.} In the following we extend the definition of an encoder given in \cite[Definition 3.2]{KviaDP} so that it is able to deal with packing-certifiable problems. The main difference is that now the function $\f$ is incorporated in the definition of an encoder, since as discussed in the introduction we need to consider an additional scenario where the entries of the table are discarded (technically, this is modeled by setting those entries to ``$-\infty$'') and for this we will have to deal with the partial solutions particular to each problem. In the applications of the next sections, we will call such functions  that propagate the entries to be discarded \emph{relevant}. We also need to add a condition about the {\sl computability} of the function \f, so that encoders can indeed be used for performing dynamic programming.

%


\begin{definition} \label{defi:encod}
An \emph{encoder} is a triple $\E = (\C,\L,\f)$ where
\begin{itemize}
\item [\C] is a function in $2^\mathbb{N} \to 2^{\Upsilon^*}$ that maps a finite subset of integers $ I \subseteq \mathbb{N}$ to a set $\C(I)$ of strings over some alphabet $\Upsilon$. Each string $R \in \C(I)$ is called an \emph{encoding}. The \emph{size} of the encoder is the function $s_{\E} : \mathbb{N} \to \mathbb{N}$ defined as $s_{\E}(t)  :=  \max \{|\C(I)| : I \subseteq \{1,\ldots,t\}\} $, where $|\C(I)|$ denotes the number of encodings in  $\C(I)$;

\item [\L] is a computable language which accepts triples $(G,{\cal S},R)\in \Gamma^* \times \Sigma^* \times \Upsilon^*$, where $G$ is a boundaried graph, ${\cal S}$ is a collection of subgraphs of $G$ and $R \in \C(\Lambda(G))$ is an encoding.
If $(G,{\cal S},R)\in \L$, we say that ${\cal S}$ \emph{satisfies} the encoding $R$ in $G$; and

\item [\f] is a  computable function in $\Gamma^* \times \Upsilon^* \to \mathbb{N}\cup\{ -\infty\}$ that maps a boundaried graph $G$ and an encoding $R\in \C(\Lambda(G))$ to an integer or to $-\infty$.

\end{itemize}


\end{definition}



As it will become clear with the applications described in the next sections, an encoder is a formalization of the tables used by an algorithm that solves a packing-certifiable problem $\Pi$ by doing DP over a tree decomposition of the input graph. The encodings in $\C(I)$ correspond to the entries of the DP-tables of graphs with boundary labeled by the set of integers $I$. The language \L identifies certificates which are partial solutions satisfying the boundary conditions imposed by an encoding.


The following definition differs from~\cite[Definition 3.3]{KviaDP} as now the function \f is incorporated in the definition of an encoder \E.


\begin{definition} \label{defi:pi-encod}
Let $\Pi$ be a packing-certifiable problem. An encoder \E is a \emph{$\Pi$-encoder} if $\C(\emptyset)$ is a singleton, denoted by $\{R_\emptyset\}$, such that for any $0$-boundaried graph $G$, $\f(G,R_\emptyset) = f^\Pi(G)$.
\end{definition}
%
%


The following definition allows to control the number of possible distinct values assigned to encodings and plays a similar role to FII or \emph{monotonicity} in previous work~\cite{BFL+09,KLP+12,FLST10}.


\begin{definition} \label{def:confined}
An encoder $\E$ is \emph{$g$-confined} if there exists a function $g : \mathbb{N} \to \mathbb{N}$ such that for any $t$-boundaried graph $G$ with $\Lambda(G) = I$ it holds that either $\{R \in \C(I) : \f(G,R) \neq - \infty \}= \emptyset\ $ or
$\ \max_{R} \{\f(G,R)\neq - \infty \}\  -\   \min_{R} \{\f(G,R) \neq - \infty\}  \  \leq \ g(t)$.

\end{definition}

For an encoder \E and a function $g$, in the next sections we will denote the {\sl relevant} functions discussed before by  $\fbar$ to distinguish them from other functions that we will need.




\smallskip\smallskip
\noindent \textbf{Equivalence relations and representatives.} We now define some equivalence relations on $t$-boundaried graphs.

\begin{definition} \label{defi:equiv}
Let \E be an encoder, let $G_1,G_2 \in \B$, and let $\G$ be a class of graphs.
\begin{enumerate}
\item $G_1 \eqs\E G_2$ if
$\Lambda(G_1)=\Lambda(G_2)=: I$ and
there exists an integer $\D\E(G_1,G_2)$ (depending on $G_1, G_2$) such that for any encoding $R   \in \C(I)$ we have $\f(G_1,R) = \f(G_2,R) - \D\E (G_1,G_2)$.

\item $G_1 \eq\G G_2$ if either $G_1 \notin \G$ and $G_2 \notin \G$, or
$G_1,G_2 \in \G$ and,
for any $H \in \B$, $H \oplus G_1 \in \G$ if and only if $H \oplus G_2 \in \G$.

\item $G_1 \eqsG\E G_2$ if
$G_1 \eqs\E G_2$ and
$G_1 \eq\G G_2$.

\item If we restrict the graphs $G_1 , G_2$ to be in \F, then the corresponding equivalence relations, which are a restriction of \eqs\E and \eqsG\E, are denoted by \eq\E and \eqG\E, respectively.
\end{enumerate}
\end{definition}

If for all encodings $R$, $\f(G_1,R) = \f(G_2,R) = -\infty$, then we set $\D\E (G_1,G_2):=0$ (note that any fixed integer would satisfy the first condition in Definition~\ref{defi:equiv}). Following the notation of Bodlaender \emph{et al}.~\cite{BFL+09}, the function $\D\E$ is called the \emph{transposition function} for the equivalence relation \eqs\E. Note that we can use the restriction of $\D\E$ to couples of graphs in \F to define the equivalence relation \eq\E.


In the following we only consider classes of graphs whose membership can be expressed in Monadic Second Order (MSO) logic. Therefore, we know that the number of equivalence classes of \eq\G is finite~\cite{Buc60}, say at most $r_{\G,t}$, and we can state the following lemma.

\begin{lemma}$[\star]$ \label{lem:nb class}
Let \G be a class of graphs whose membership is expressible in MSO logic.
For any encoder \Enospace, any function $g: \mathbb{N} \rightarrow \mathbb{N}$ and any integer $t \in \mathbb{N}$,
if \E is $g$-confined
then the equivalence relation \eqsG\E has at most  $ r(\E,g,t,\G):=(g(t)+2)^{s_{\E}(t)} \cdot 2^t \cdot r_{\G,t} $ equivalence classes. In particular, the equivalence relation \eqG\E has at most $ r(\E,g,t,\G)$ equivalence classes as well.
\end{lemma}

\begin{definition} \label{defi:DP-friend}
An equivalence relation \eqsG\E is \emph{DP-friendly} if, for any graph $G \in \B$ with $\partial(G) = A$ and any two boundaried graphs $H$ and $G_B$ with $G= H \oplus G_B$ such that $G_B$ has boundary $B \subseteq V(G)$ with $|B| \leq t$ and $A\cap V(G_B) \subseteq B$, the following holds.
Let $G'\in \B$ with $\partial(G') = A$ be the graph obtained from $G$ by replacing the subgraph $G_B$ with some $G_B' \in \B$ such that $G_B \eqsG\E G_B'$. Then $G \eqsG\E G'$ and $\D\E(G,G') = \D\E(G_B, G_B')$.
\end{definition}


The following useful fact  states that for proving that \eqsG\E is DP-friendly, it suffices to prove that $G \eqs\E G'$ instead of $G \eqsG\E G'$.


\begin{fact}$[\star]$ \label{fait:equiv}
Let $G\in \B$ with a separator $B$, let $G_B \eqG\E G_B'$, and let $G' \in \B$ as in Definition~\ref{defi:DP-friend}. If $G \eqs\E G'$, then $G \eqsG\E G'$.
\end{fact}


In order to perform a protrusion replacement that does not modify the behavior of the graph with respect to a problem $\Pi$, we need the relation \eqs\E to be a refinement of the canonical equivalence relation $\equiv_{\Pi,t}$.


\begin{lemma}$[\star]$\label{lem:refine eq}
Let $\Pi$ be a packing-certifiable parameterized problem defined on a graph class ${\cal G}$, let \E be an encoder, let $g: \mathbb{N} \rightarrow \mathbb{N}$, and let $G_1, G_2 \in \B$. If \E is a $g$-confined $\Pi$-encoder and \eqsG\E is DP-friendly, then the fact that $G_1 \eqsG\E G_2$ implies the following:
\begin{itemize}
\item[$\bullet$] $G_1 \equiv_\Pi G_2$; and
\item[$\bullet$] $\Delta_{\Pi,t}(G_1,G_2) = \D\E(G_1,G_2)$.
\end{itemize}
In particular, this holds when $G_1, G_2 \in \F$  and $G_1 \eqG\E G_2$.
\end{lemma}

\begin{definition} \label{defi:progres}
Given an encoder \E and an equivalence class $\mathfrak{C} \subseteq \F$ of \eqG\E, a graph $G \in \mathfrak{C}$ is a \emph{progressive representative} of $\mathfrak{C}$ if for any $G'\in \mathfrak{C}$, it holds that $\D\E(G,G') \leq 0$.
\end{definition}


\begin{lemma}$[\star]$\label{lem:progres size}
Let \G be a class of graphs whose membership is expressible in MSO logic. For any encoder \E, any function $g: \mathbb{N} \rightarrow \mathbb{N}$, and any $t \in \mathbb{N}$, if \E is $g$-confined and \eqsG\E is DP-friendly, then any equivalence class of \eqG\E has a progressive representative of size at most $b(\E,g,t,\G) := 2^{r(\E,g,t,\G)+1} \cdot t$, where $r(\E,g,t,\G)$ is the function defined in Lemma~\ref{lem:nb class}.
\end{lemma}

\smallskip
\noindent \textbf{An explicit protrusion replacement.} 
%
%
The next lemma specifies conditions under which, given an upper bound on the size of the representatives, a generic DP algorithm can provide in linear time an explicit protrusion replacer.

\begin{lemma}$[\star]$ \label{lem:comput repres}
Let \G be a class of graphs, let \E be an encoder, let $g: \mathbb{N} \rightarrow \mathbb{N}$,  and let $t \in \mathbb{N}$ such that \E is $g$-confined  and \eqsG\E is DP-friendly.
Assume we are given an upper bound $b \geq t$ on the size of a smallest progressive representative of any class of \eqG\E. Given a $t$-protrusion $Y$ inside some graph, we can compute a $t$-protrusion $Y'$ of size at most $b$ such that $Y \eqG\E Y'$ and $\D\E(Y',Y) \leq 0$.
Furthermore, such a protrusion can be computed in time $ O(|Y|)$, where the hidden constant depends only on $\E,g,b,\G$, and $t$.
\end{lemma}

Let us now piece everything together to state the main result of~\cite{KviaDP} that we need to reprove here for packing-certifiable problems.
For issues of constructibility, we restrict \G to be the class of $H$-(topological)-minor-free graphs.

\begin{theorem}$[\star]$   \label{theo:main}
Let \G be the class of graphs excluding some fixed graph $H$ as a (topological) minor and let $\Pi$ be a parameterized packing-certifiable problem defined on \G. Let \E be an encoder, let $g: \mathbb{N} \rightarrow \mathbb{N}$, and let $t \in \mathbb{N}$ such that \E is a $g$-confined $\Pi$-encoder and \eqsG\E is DP-friendly. Given an instance $(G,k)$ of $\Pi$ and a $t$-protrusion $Y$ in $G$, we can compute in time $O(|Y|)$ an equivalent instance $(G-(Y-\partial(Y)) \oplus Y',k')$ where $Y'$ is a $t$-protrusion with $|Y'| \leq b(\E,g,t,\G)$ and $k'\leq k$ and where $b(\E,g,t,\G)$ is the function defined in Lemma~\ref{lem:progres size}.
\end{theorem}



Such a protrusion replacer can be used to obtain a kernel when, for instance, one is able to provide a protrusion decomposition of the instance.

\begin{corollary}$[\star]$   \label{coro:main}
Let \G be the class of graphs excluding some fixed graph $H$ as a (topological) minor and let $\Pi$ be a parameterized packing-certifiable problem defined on \G. Let \E be an encoder, let $g: \mathbb{N} \rightarrow \mathbb{N}$, and let $t \in \mathbb{N}$ such that \E is a $g$-confined $\Pi$-encoder and \eqsG\E is DP-friendly. Given an instance $(G,k)$ of $\Pi$ and an $(\alpha k,t)$-protrusion decomposition of $G$, we can construct a linear kernel for $\Pi$ of size at most $ (1+b(\E,g,t,\G))\cdot \alpha \cdot k$, where $b(\E,g,t,\G)$ is the function defined in Lemma~\ref{lem:progres size}.
\end{corollary}


\section{Main ideas for the applications}
\label{sec:applications}

In this section by sketch  the main ingredients that we use in our applications for obtaining the linear kernels, before going through the details for each problem in the next sections.




\vspace{.2cm}
\noindent\textbf{\textsc{General methodology}.} The next theorem will be fundamental in the applications.

\begin{theorem}[Kim \emph{et al}.~\cite{KLP+12}]\label{theo:prot dec}
Let $c,t$ be two positive integers, let $H$ be an $h$-vertex graph, let $G$ be an $n$-vertex $H$-topological-minor-free graph, and let $k$ be a positive integer. If we are given a set $X \subseteq V(G)$ with $|X| \leq c \cdot k$ such that $\tw(G-X) \leq t$, then we can compute in time $O(n)$ an $((\alpha_{H} \cdot t \cdot c)\cdot k, 2t + h)$-protrusion decomposition of $G$, where $\alpha_{H}$ is a constant depending only on $H$, which is upper-bounded by $40 h^2 2 ^{5 h \log h}$.
\end{theorem}

A typical application of our framework for obtaining an explicit linear kernel for a packing-certifiable problem $\Pi$ on a graph class \G is as follows. The first task is to define an encoder $\E$ and to prove that for some function $g: \mathbb{N} \rightarrow \mathbb{N}$, \E is a $g$-confined $\Pi$-encoder and \eqsG\E is DP-friendly.
The next ingredient is a  polynomial-time algorithm that, given an instance $(G,k)$ of $\Pi$, either reports that $(G,k)$ is a \textsc{Yes}-instance (or a \textsc{No}-instance, depending on the problem), or finds a treewidth-modulator of $G$ with size $O(k)$. The way to obtain this algorithm depends on each particular problem and in our applications we will use a number of existing results in the literature in order to find it. Once we have such a linear treewidth-modulator, we can use Theorem~\ref{theo:prot dec} to find a linear protrusion decomposition of $G$. Finally, it just remains to apply Corollary~\ref{coro:main} to obtain an explicit linear kernel for $\Pi$ on $\G$; see Figure~\ref{fig:scheme} for a schematic illustration.

\vspace{.25cm}
\begin{figure}[h!]
\vspace{-.1cm}
\begin{center}
\scalebox{.86}{\includegraphics{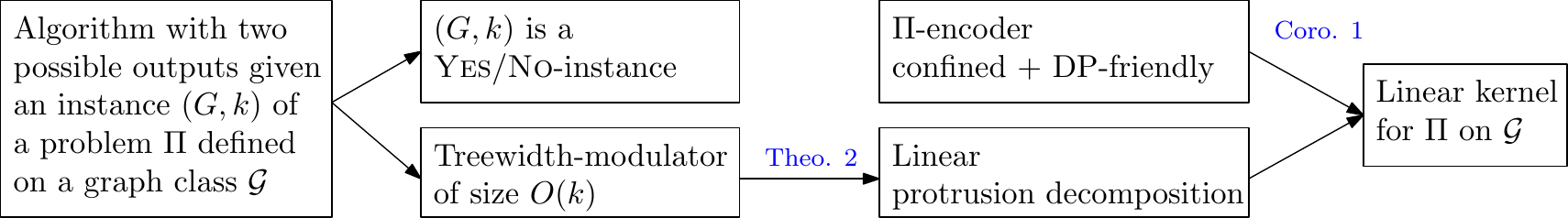}}
\end{center}\vspace{-0.3cm}
\caption{Illustration of a typical application of the framework presented in this article.}
\label{fig:scheme}
\end{figure}

Let us provide here some generic  intuition about the additional criterion mentioned in the introduction to discard the entries in the tables of an encoder. For an encoder $\E = (\C,\L,\f)$ and a function $g: \mathbb{N} \to \mathbb{N}$, we need some notation in order to define the \emph{relevant function} \fbar, which will be an appropriate modification of \f. Let $G \in \B$ with boundary $A$ and let $R_A$ be an encoding. We (recursively) define $R_A$ to be \emph{irrelevant for \fbar} if there exists a certificate $\mc{S}$ such that $(G,\mc{S},R_A)\in \L$ and $|\mc{S}|= \f(G,R_A)$ and a separator $B \subseteq V(G)$ with $|B| \leq t$ and $B \neq A$, such that $\mc{S}$ {\sl induces} an encoding $R_B$ in the graph $G_B \in \B$ with $ \fbar(G_B,R_B) = -\infty$. Here, by using the term ``induces'' we implicitly assume that $\mc{S}$ defines an encoding $R_B$ in the graph $G_B$; this will be the case in all the encoders used in our applications.

To define \fbar, we will always use the following {\sl natural} function \f, which for each problem $\Pi$ is meant to correspond to an extension to boundaried graphs of the maximization function $f^{\Pi}$ of Definition~\ref{defi:optimization-function}. For a graph $G$ and an encoding $R$, this natural function is defined as $\f(G,R)  =  \max \{k  :  \exists \mc{S}, |\mc{S}| \geq k, (G,\mc{S},R) \in \L\}$. Then we define the function $\fbar$ as follows:

\begin{equation*} \label{eq: relevant f}
\fbar(G,R_A) =\
\left\{\begin{array}{lll}
  & -\infty,  & \text{if } \f(G,R_A) + g(t) < \max \{\f(G,R): R \in \C(\Lambda(G)) \}, \\
  &     &     \text{or if $R_A$ is irrelevant for \fbar.} \\
%
  & \f(G,R_A),  & \mbox{otherwise}.\\
\end{array}\right.
\end{equation*}

That is, we will use the modified encoder $(\C,\L,\fbar)$. We need to guarantee that the above function \fbar is computable, as required\footnote{The fact that the values of the function \fbar can be calculated is important, in particular, in the proof of Lemma~\ref{lem:comput repres}, since we need to be able to compute equivalence classes of the equivalence relation \eqG\E.} in Definition~\ref{defi:encod}. Indeed, from the definition it follows that an encoding $R_A$ defined at a node $x$ of a given tree decomposition is irrelevant if and only if $R_A$ can be obtained by combining encodings corresponding to the children of $x$, such that at least one of them is irrelevant. This latter property can be easily computed recursively on a tree decomposition, by performing standard dynamic programming. We will omit this computability issue in the applications, as the same argument sketched here applies to all of them.


In order to  obtain the linear treewidth-modulators mentioned before, we will use several results from~\cite{BFL+09,FLST10,FLRS10}, which in turn use the following two propositions. For an integer  $r\geq 2$, let ${\rm \Gamma}_{r}$  be the graph obtained from the  $(r\times r)$-grid by
triangulating internal faces such that all internal vertices become  of degree $6$,
all non-corner external vertices are of degree 4,
and  one corner of degree 2 is made adjacent to all vertices
of the external face (the {\em corners} are the vertices that in the underlying grid have degree 2).
As an example, the graph $\Gamma_6$ is shown in Figure~\ref{fig-gamma-k}.

\begin{figure}[htb]
\begin{center}
\scalebox{.8}{\includegraphics{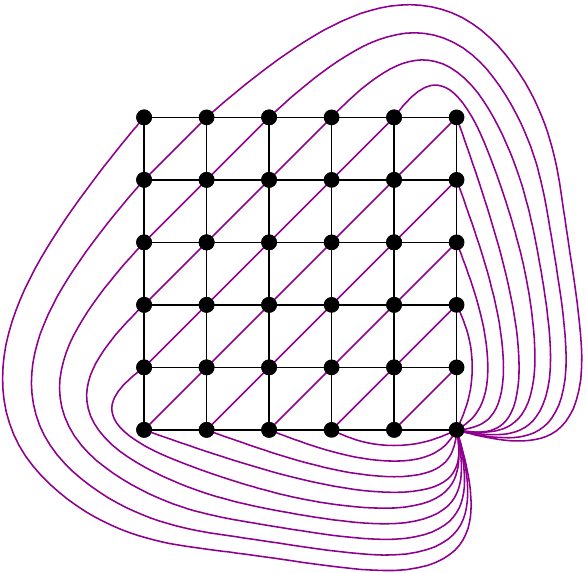}}
\end{center}\vspace{-.15cm}
\caption{The graph $\Gamma_{6}$.}
\label{fig-gamma-k}
\end{figure}


\begin{proposition}[Demaine  and Hajiaghayi~\cite{DH08}]\label{prop:tw-minor}
There is a function $f_m:\mathbb{N}\rightarrow\mathbb{N}$ such that for every $h$-vertex graph $H$
and every positive integer $r$, every $H$-minor-free graph with treewidth at least $f_{m}(h)\cdot r$, contains
the $(r\times r)$-grid as a minor.
\end{proposition}

\begin{proposition}[Fomin {\em et al.} \cite{FGT11}]\label{prop:tw-contraction}
There is a function $f_c:\mathbb{N}\rightarrow\mathbb{N}$ such that
for every $h$-vertex apex graph $H$ and every positive integer $r$,
every $H$-minor-free graph with treewidth at least $f_{c}(h)\cdot r$,
contains the graph  ${\rm \Gamma}_{r}$ as a contraction.
\end{proposition}

The current best upper bound~\cite{KaKo12} for the function $f_m$ is $f_m (h) = 2^{ O(h^2  \log h)}$ and, up to date, there is no explicit bound for the function $f_c$. We would like to note that this non-existence of explicit bounds for $f_c$ is an issue that concerns the {\sl graph class} of $H$-minor-free graphs and it is perfectly compatible with our objective of providing explicit constants for particular {\sl problems} defined on that graph class, which will depend on the function $f_c$.


Let us now provide a sketch of the main basic ingredients used in each of the applications.

\vspace{.25cm}
\noindent\textbf{Packing minors.} Let $\mc{F}$ be a fixed finite set of graphs. In the \textsc{$\mc{F}$-Packing} problem, we are given a graph $G$ and an integer parameter $k$ and the question is whether $G$ has $k$ vertex-disjoint subgraphs $G_1,\ldots,G_k$,
                     each containing some graph in $\mc{F}$ as a minor.
%
%
When all the graphs in $\mc{F}$ are connected and $\mc{F}$ contains at least one planar graph, we call the problem  \textsc{Connected-Planar-$\mc{F}$-Packing}. The encoder uses the notion of rooted packing introduced by Adler \emph{et al}.~\cite{ADF+11}, which we also used in~\cite{KviaDP} for \textsc{Connected-Planar-$\mc{F}$-Deletion}. To obtain the treewidth-modulator, we use the  \emph{Erd\H{o}s-P\'{o}sa property} for graph minors~\cite{ErPo65,RoSe86,ChChSTOC13}. More precisely, we use that on minor-free graphs, as proved by Fomin \emph{et al}.~\cite{FST11}, if $(G,k)$ is a \textsc{No}-instance of \textsc{Connected-Planar-$\mc{F}$-Packing}, then $(G,k')$ is a \textsc{Yes}-instance of \textsc{Connected-Planar-$\mc{F}$-Deletion} for $k' = O(k)$. Finally, we use a result of Fomin \emph{et al}.~\cite{FLST10} that provides a polynomial-time algorithm to find treewidth-modulators for \textsc{Yes}-instances of \textsc{Connected-Planar-$\mc{F}$-Deletion}. The obtained constants involve, in particular, the currently best known constant-factor approximation of treewidth on minor-free graphs.

\vspace{.25cm}
\noindent\textbf{Packing scattered minors.} Let  $\mc{F}$ be a fixed finite set of graphs and let $\ell$ be a positive integer.   In the $\ell$-$\mc{F}$-Packing problem, we are given a graph $G$ and an integer parameter $k$ and the question is whether $G$ has $k$ subgraphs $G_1,\ldots,G_k$ pairwise at distance at least $\ell$,
                     each containing some graph from $\mc{F}$ as a minor.
%
%
The encoder for \rFPack is a combination of the encoder for \FPack and the one for $\ell$-\textsc{Scattered Set} that we used in~\cite{KviaDP}. For obtaining the treewidth-modulator, unfortunately we cannot proceed as for packing minors, as up to date no linear Erd\H{o}s-P\'{o}sa property for packing scattered planar minors is known; the best bound we are aware of is $O(k\sqrt{k})$, which is {\sl not} enough to obtain a linear kernel. To circumvent this problem, we use the following trick: we (artificially) formulate \rFPack as a vertex-certifiable problem and prove that it fits the conditions required by the framework of Fomin~\emph{et al}.~\cite{FLST10} to produce a treewidth-modulator. (We would like to stress that this formulation of the problem as a vertex-certifiable one is {\sl not} enough to apply the results of~\cite{KviaDP}, as one has to further verify the necessary properties of the encoder are satisfied and it does not seem to be an easy task at all.) Once we have it, we consider the original formulation of the problem to define its encoder. As a drawback of resorting to the general results of~\cite{FLST10} and, due to the fact that \rFPack is contraction-bidimensional, we provide linear kernels for the problem on the (smaller) class of apex-minor-free graphs. 

\vspace{.25cm}
\noindent\textbf{Packing overlapping minors.}  Let  $\mc{F}$ be a fixed finite set of graphs and let $\ell$ be a positive integer. In the \textsc{$\cal F$-Packing with $\ell$-Membership} problem, we are given a graph $G$ and an integer parameter $k$ and the question is whether $G$ has $k$ subgraphs $ G_1,\dots,G_k$ such that
 each subgraph contains some graph from $\cal F$ as a minor,
                     and each vertex of $G$ belongs to at most $\ell$ subgraphs.
%
The encoder is an enhanced version of the one for packing minors, in which we allow a vertex to belong simultaneously to several minor-models. To obtain the treewidth-modulator,  the situation is simpler than above, thanks to the fact that a packing of models is in particular a packing of models with $\ell$-membership. This allows us to use the linear Erd\H{o}s-P\'{o}sa property that we described for packing minors and therefore to construct linear kernels on minor-free graphs.

\vspace{.25cm}
\noindent\textbf{Packing scattered and overlapping subgraphs.} The definitions of the corresponding problems are similar to the ones above, just by replacing the minor  by the subgraph relation. The encoders are simplified versions of those that we defined for packing scattered and overlapping minors, respectively. The idea for obtaining the treewidth-modulator is to apply a simple reduction rule that removes all vertices not belonging to any of the copies of the subgraphs we are looking for. It can be easily proved that if a reduced graph is a \textsc{No}-instance of the problem, then it is a \textsc{Yes}-instance of $\ell'$-\textsc{Dominating Set},
where $\ell'$ is a function of the integer $\ell$ corresponding to the problem and the largest diameter of a subgraph in the given family. We are now in position to use the machinery of~\cite{FLST10} for $\ell'$-\textsc{Dominating Set} and find a linear treewidth-modulator.


\renewcommand{\Epi}[2]{\ensuremath{\mc{#1 E}_{\!\mc{F}\!\sc{P}}^{#2}}\xspace}
\renewcommand{\Ebar}{{\Epi{\bar}{}}\xspace}
\renewcommand{\fbar}{\ensuremath{\bar f^{\Epi{}{}}_g}\xspace}

\section{A linear kernel for  \textsc{Connected-Planar-$\mc{F}$-Packing}} \label{sec: FPack}

Let $\mc{F}$ be a finite set of graphs. We define the \textsc{$\mc{F}$-Packing} problem as follows.

\vspace{.4cm}
\begin{boxedminipage}{.9\textwidth}
\textsc{$\mc{F}$-Packing}
\vspace{.1cm}

\begin{tabular}{ r l }
\textbf{Instance:}  & A graph $G$ and a non-negative integer $k$. \\
\textbf{Parameter:} & The integer $k$.\\
\textbf{Question:}  & Does $G$ have $k$ vertex-disjoint subgraphs $G_1,\ldots,G_k$\\
                    & ~~~each containing some graph in $\mc{F}$ as a minor?\\
\end{tabular}
\end{boxedminipage}
\vspace{.4cm}

In order to build a protrusion decomposition for instances of the above problem, we use a version of the Erd\H{o}s-P\'{o}sa property (see Definition~\ref{defi: Erdos Posas} and Theorem~\ref{theo: Erdos Posa}) 
 that establishes a linear relation between \NO-instances of \FPack and \YES-instances of \textsc{$\mc{F}$-Deletion}, and then we apply tools of Bidimensionality theory on \textsc{$\mc{F}$-Deletion} (see Corollary~\ref{coro: tw modul FDel}).
Hence, we also need to define the \textsc{$\mc{F}$-Deletion} problem.

\vspace{.4cm}
\begin{boxedminipage}{.9\textwidth}
\textsc{$\mc{F}$-Deletion}
\vspace{.1cm}

\begin{tabular}{ r l }
\textbf{Instance:} & A graph $G$  and a non-negative integer $k$. \\
\textbf{Parameter:} & The integer $k$.\\
\textbf{Question:} & Does $G$ have a set $S\subseteq V(G)$
such that $|S|\leqslant k$\\ &~~~and $G-S$ is $H$-minor-free for every $H\in
\mc{F}$?\\
\end{tabular}
\end{boxedminipage}
\vspace{.4cm}

When all the graphs in $\mc{F}$ are connected, the corresponding problems are called  \textsc{Connected-$\mc{F}$-Packing} and \textsc{Connected-$\mc{F}$-Deletion},
and when $\mc{F}$ contains at least one planar graph, we call them  \textsc{Planar-$\mc{F}$-Packing} and \textsc{Planar-$\mc{F}$-Deletion}, respectively. When both conditions are satisfied, the problems are called  \textsc{Connected-Planar-$\mc{F}$-Packing} and \textsc{Connected-Planar-$\mc{F}$-Deletion} (the parameterized versions of these problems are respectively denoted by $\textsc{c}\mc{F}\textsc{P}$, $\textsc{c}\mc{F}\textsc{D}$, $\textsc{p}\mc{F}\textsc{P}$, $\textsc{p}\mc{F}\textsc{D}$, $\textsc{cp}\mc{F}\textsc{P}$, and $\textsc{cp}\mc{F}\textsc{D}$).


\smallskip

In this section we present a linear kernel for \textsc{Connected-Planar-$\mc{F}$-Packing} on the family of graphs excluding a fixed graph $H$ as a minor.

\smallskip

We need to define which kind of structure a certificate for \FPack is. For an arbitrary graph, a solution will consist of a \emph{packing of models} as defined below. We also recall the definition of model.

\begin{definition} \label{defi: model}
A \emph{model} of a graph $F$ in a graph $G$ is a mapping $\Phi$ that assigns
to every vertex $v \in V(F)$ a non-empty connected subgraph $\Phi(v)$ of $G$, and 
to every edge $e \in E(F)$ an edge $\Phi(e) \in E(G)$, such that:
\begin{itemize}
       \item[$\bullet$] the graphs $\Phi(v)$ for $v \in V(F)$ are mutually vertex-disjoint and
             the edges $\Phi(e)$ for $e \in E(F)$ are pairwise distinct;
       \item[$\bullet$] for $\{u,v\} \in E(F)$, $\Phi(\{u,v\})$ has one endpoint in $V(\Phi(u))$ and the other in $V(\Phi(v))$.
\end{itemize}

We denote by $\Phi(F)$ the subgraph of $G$ obtained by the (disjoint) union of the subgraphs $\Phi(v)$ for $v\in V(F)$ plus the edges $\Phi(e)$ for $e \in E(F)$.
\end{definition}

\begin{definition} \label{defi: packing models}

Given a set $\mc{F}$ of minors and a graph $G$, a \emph{packing of models} $\mc{S}$ is a set of vertex-disjoint models. That is, the graphs $\Phi(F)$ for $\Phi \in \mc{S}, F \in  \mc{F}$ are pairwise vertex-disjoint.

\end{definition}

\subsection{A protrusion decomposition for an instance of \FPack}\label{ssec: FPack prot decompo}

In order to find a linear protrusion decomposition, we need some preliminaries.
\begin{definition}\label{defi: Erdos Posas}
A class of graphs $\mc{F}$ satisfies the \emph{Erd\H{o}s-P\'{o}sa property} \cite{ErPo65} if
   there exists a function $f$ such that, for every integer $k$ and every graph $G$, either $G$ contains $k$ vertex-disjoint subgraphs each isomorphic to a graph in $\mc{F}$,
   or there is a set $S \subseteq V(G)$ of at most $f(k)$ vertices such that $G - S$ has no subgraph in $\mc{F}$.
\end{definition}

Given a connected graph $F$, let $\mc{M}(F)$ be the class of graphs that can be contracted to $F$. Robertson and Seymour \cite{RoSe86} proved that $\mc{M}(F)$ satisfies the Erd\H{o}s-P\'{o}sa property if and only if $F$ is planar.
A significant improvement on the function $f(k)$ has been recently provided by Chekuri and Chuzhoy \cite{ChChSTOC13}. When $G$ belongs to a proper minor-closed family, Fomin \emph{et al}.~\cite{FST11} proved that $f$ can be taken to be linear for any planar graph $F$. It is not difficult to see that these results also hold if instead of a connected planar graph $F$, we consider a finite family $\mc{F}$ of connected graphs containing at least one planar graph. This discussion can be summarized as follows, with a precise upper bound on the desired linear constant.

\begin{theorem}[Fomin \emph{et al}.~\cite{FST11}]\label{theo: Erdos Posa}
Let $\mc{F}$ be a finite family of connected graphs containing at least one planar graph on $r$ vertices, let $H$ be an $h$-vertex graph, and let \G be the class of $H$-minor-free graphs. There exists a constant $c$ such that if $(G,k)\notin \textsc{cp}\mc{F}\textsc{P}_\G$,
then $(G, c\cdot r\cdot 2^{15h+8h \log h} \cdot k) \in  \textsc{cp}\mc{F}\textsc{D}_\G$.
\end{theorem}

The next theorem provides a way to find a treewidth-modulator for an instance of a problem verifying the so-called \emph{bidimensionality} and \emph{separability} properties restricted to the class of (apex)-minor-free graphs.
Loosely speaking, the algorithm consists in building a tree decomposition of the instance, then finding a bag that separates the instance in such a way that the solution is balanced, and finally finding recursively other bags in the two new tree decompositions.
In order to make the algorithm constructive, we need to build a tree decomposition of the input graph whose width differs from the optimal one by a constant factor. To this aim, we use a (polynomial) approximation algorithm of treewidth on minor-free graphs, which is well-known to exist. Let us denote by $\tau_H$ this approximation ratio. To the best of our knowledge there is no explicit upper bound on this ratio, but one can be derived from the proofs of Demaine and Hajiaghayi~\cite{DH08}. We note that any improvement on this constant will directly translate to the size of our kernels. We also need to compute an initial solution of the problem under consideration. Fortunately, for all our applications, there is an {\sc{EPTAS}} on minor-free graphs \cite{FLRS10}. By choosing the approximation ratio of the solution to be $2$, we can announce the following theorem adapted from Fomin \emph{et al}.~\cite{FLST10}.


\begin{theorem}[Fomin \emph{et al}.~\cite{FLST10}] \label{theo: tw mod}
For any real $\epsilon > 0 $ and any minor-bidimensional (resp. contraction-bidimensional) linear-separable problem $\Pi$  on the class \G of graphs that exclude a minor $H$ (resp. an apex-minor $H$), there exists an integer $t \geq 0$ such that any graph $G\in\G$ has a treewidth-$t$-modulator of size at most $\epsilon\cdot f^\Pi(G)$.
\end{theorem}

The impact of the tree decomposition approximation is hidden in the value of $t$, and the impact of the solution approximation will be hidden in the ``$O$'' notation. The parameters from the class of graphs or from the problem  will affect the {\sl time complexity} of the algorithm, and not the size of our kernel.
In our applications we state corollaries of the above result (namely, Corollary~\ref{coro: tw modul FDel} and Corollary~\ref{coro: tw modul rFpack}) in which we choose $\epsilon = 1$ and we provide an explicit bound on the value of $t$.


We are in position to state the following corollary
 claiming that, given an instance of \textsc{Planar-$\mc{F}$-Deletion}, in polynomial time we can either find a treewidth-modulator or report that is a \NO-instance. This is a corollary of the result of Fomin \emph{et al}.~\cite{FLST10} stated in Theorem \ref{theo: tw mod}, where $\epsilon$ is fixed to be $1$.
The bound on the treewidth is derived from the proof of Theorem \ref{theo: tw mod} in~\cite{FLST10}. 

\begin{corollary} \label{coro: tw modul FDel}
Let $\mc{F}$ be a finite set of graphs containing at least one $r$-vertex planar graph $F$, let $H$ be an $h$-vertex graph, and let \G be the class of $H$-minor-free graphs.
If $(G,k') \in \textsc{p}\mc{F}\textsc{D}_\G$, then there exists a set $X\subseteq V(G)$
  such that $|X|= k'$                                          
  and       $\tw(G-X)= O (r\sqrt r \cdot \tau_H^3 \cdot f_m(h)^3 )$.            
Moreover, given an instance  $(G,k)$ with $|V(G)|=n$, there is an algorithm running in time $O(n^{3})$ that either finds such a set $X$ or correctly reports that $(G,k) \notin \textsc{p}\mc{F}\textsc{D}_\G$.

%
\end{corollary}

 Note that since in Theorem \ref{theo: tw mod} the value of $\epsilon$ can be chosen arbitrarily, we can state many variants of the above corollary.
For instance, in our previous article~\cite{KviaDP}, we used the particular case where $|X|=O(r \cdot f_m(h) \cdot k')$ and $\tw(G-X)= O(r \cdot f_m(h)^2)$.

We are now able to construct a linear protrusion decomposition.

\begin{lemma} \label{lem: FPack prot decompo}
Let $\mc{F}$ be a finite set of graphs containing at least one $r$-vertex planar graph $F$, let $H$ be an $h$-vertex graph, and let \G be the class of $H$-minor-free graphs.
Let $(G,k)$ be an instance of \textsc{Connected-Planar}-\FPack . If $(G,k)  \notin \textsc{cp}\mc{F}\textsc{P}_\G$, then we can construct  in polynomial time  a linear protrusion decomposition of $G$.
\end{lemma}
\begin{proof}
Given an instance $(G,k)$ of $\textsc{cp}\mc{F}\textsc{P}_\G$, we run the algorithm given by Corollary \ref{coro: tw modul FDel} for the \textsc{Connected-Planar-$\mc{F}$-Deletion} problem with input $(G, k' = c\cdot r\cdot 2^{15h+8h \log h} \cdot k)$.
If the algorithm is not able to find a treewidth-modulator $X$ of size $|X|= k'$, then by Theorem \ref{theo: Erdos Posa} we can conclude that $(G,k) \in \textsc{cp}\mc{F}\textsc{P}_\G$.
Otherwise, we use the set $X$ as input to the algorithm given by Theorem \ref{theo:prot dec}, which outputs in linear time an
$((\alpha_{H} \cdot t)\cdot k', 2t + h)$-protrusion decomposition of $G$, where
\begin{itemize}
\item[$\bullet$] $t =O ( r\sqrt r \cdot \tau_H^3 \cdot f_m(h)^3 )$ is provided by Corollary \ref{coro: tw modul FDel}  (the bound on the treewidth);
\item[$\bullet$] $k'= O ( r \cdot 2^{O(h \log h)}\cdot k)$ is provided by Theorem \ref{theo: Erdos Posa} (the parameter of \textsc{$\mc{F}$-Deletion}); and 
\item[$\bullet$] $\alpha_H = O( h^2 2 ^{ O(h \log h) })$ is  the constant provided by Theorem \ref{theo:prot dec}.
\end{itemize}
That is, we  obtained  an $\left(O( h^2 2 ^{ O(h \log h) } \cdot  r^{5/2} \cdot \tau_H^3 \cdot f_m(h)^3 )\cdot k,O(r\sqrt r \cdot \tau_H^3 \cdot f_m(h)^3 )\right)$-protrusion decomposition of $G$, as claimed.
\end{proof}


\subsection{An encoder for \FPack} \label{ssec: FPack encod}

Our encoder \E for \FPack uses the notion of rooted packing \cite{ADF+11}, and is inspired by results on the \textsc{Cycle Packing} problem~\cite{BFL+09}.

Assume first for simplicity that $\mc{F}=\{F\}$ consists of a single connected graph $F$. Following~\cite{ADF+11}, we introduce a combinatorial object called \emph{rooted packing}. These objects are originally defined for branch decompositions, but can easily be translated to tree decompositions. Loosely speaking, rooted packings capture how \emph{potential models} of $F$ intersect the separator that the algorithm is processing. It is worth mentioning that the notion of rooted packing is related to the notion of \emph{folio} introduced by Robertson and Seymour \cite{RS95}, but more suited to dynamic programming.

\begin{definition}
Let $F$ be a connected graph. 
Given a set $B$ of boundary vertices of the input graph $G$, we define a \emph{rooted packing} of
$B$ as a quintuple $(\mc{A},S_F^*,S_F,\psi,\chi)$, where

\begin{itemize}
\item[$\bullet$] $S_F \subseteq S_F^*$ are both subsets of $V(F)$;
\item[$\bullet$] $\mc{A}$ is a (possible empty) collection of mutually disjoint non-empty subsets of $B$;
\item[$\bullet$]  $\psi: \mc{A} \to S_F$ is a surjective mapping assigning vertices of $S_F$ to the sets in $\mc{A}$; and
\item[$\bullet$] $\chi: S_F \times S_F \to \{0,1\}$ is a binary symmetric function between pairs of vertices in $S_F$.
\end{itemize}

We also define a \emph{potential model} of $F$ in $G$ \emph{matching} with $(\mc{A},S_F^*,S_F,\psi,\chi)$ as a partial mapping $\Phi$, that assigns
to every vertex $v \in S_F$ a non-empty subgraph $\Phi(v) \subseteq G$ such that $\{A \in \mc{A} : \psi(A)=v \}$ is the set of intersections of $B$ with connected components of $\Phi(v)$;
to every vertex $v \in S_F^* \setminus S_F$ a non-empty connected subgraph $\Phi(v) \subseteq G$; and
to every edge $e \in \{e \in E(F) : \chi(e)=1 \vee e \in S_F^* \times S_F^* \setminus S_F \}$ an edge $\Phi(e) \in E(G)$, such that $\Phi$ satisfies the two following conditions  (as in Definition \ref{defi: model}):

\begin{itemize}
       \item[$\bullet$] the graphs $\Phi(v)$ for $v \in V(F)$ are mutually vertex-disjoint and
             the edges $\Phi(e)$ for $e \in E(F)$ are pairwise distinct; and
       \item[$\bullet$] for $\{u,v\} \in E(F)$, $\Phi(\{u,v\})$ has one endpoint in $V(\Phi(u))$ and the other in $V(\Phi(v))$.
\end{itemize}

\end{definition}

\begin{figure}[h!b]
\begin{center}
\scalebox{.95}{\includegraphics{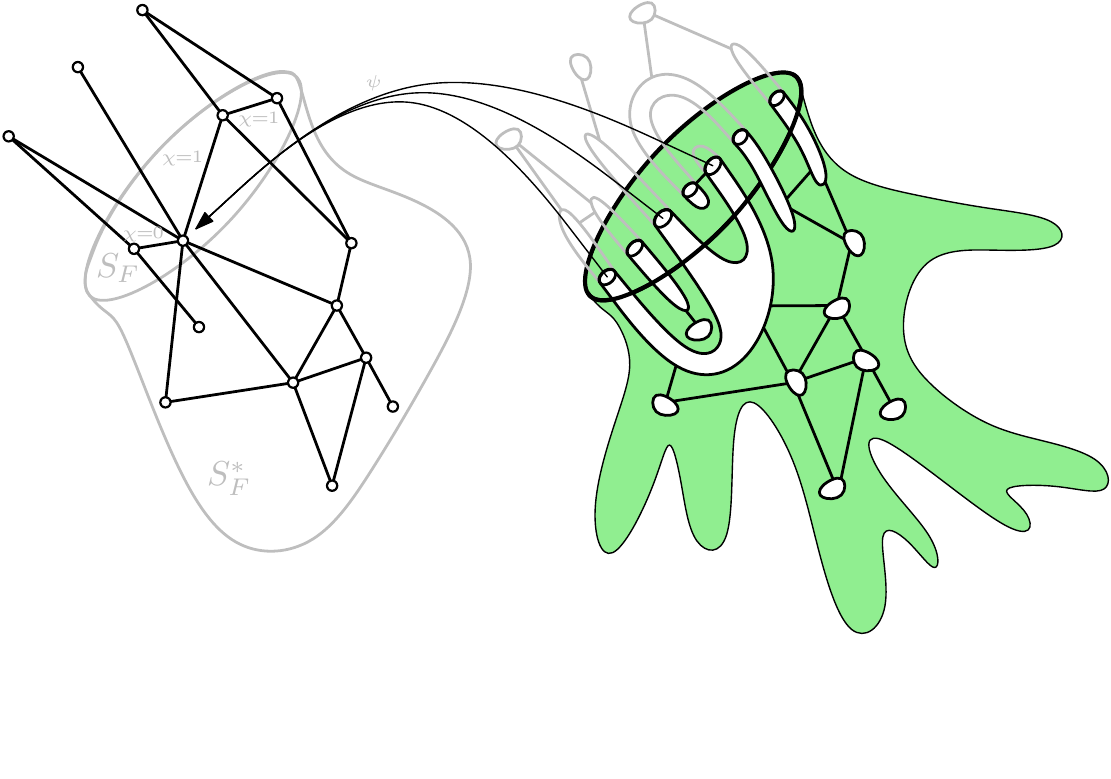}}
\end{center}\vspace{-1.3cm}
\caption{Example of a rooted packing (left) and a potential model matching with it (right).}
\label{fig:rooted-packing}
\end{figure}

See Figure~\ref{fig:rooted-packing} for a schematic illustration of the above definition. The intended meaning of a rooted packing $(\mc{A},S_F^*,S_F,\psi,\chi)$ on a separator $B$ is as follows.
The packing $\mc{A}$ represents the intersection of the connected components of the potential model with $B$.
The subsets $S_F^*,S_F \subseteq V(F)$ and the function $\chi$ indicate that we are looking in the graph $G$ for a potential model of $F[S_F^*]$ containing the edges between vertices in $S_F$ given by the function $\chi$. Namely, the function $\chi$ captures which edges of $F[S_F^*]$ have been
realized so far in the processed graph. Since we allow the vertex-models intersecting $B$ to be disconnected, we need to keep track of their connected components. The subset $S_F \subseteq S_F^*$ tells us which vertex-models intersect $B$ (in other words, $S_F$ is the boundary of $F[S_F^*]$), and the function $\psi$ associates the sets in ${\cal A}$ with the vertices in $S_F$. We
can think of $\psi$ as a coloring that colors the subsets in ${\cal A}$ with
colors given by the vertices in $S_F$. Note that several subsets in ${\cal A}$
can have the same color $u \in S_F$, which means that the vertex-model of $u$ in
$G$ is not connected yet, but it may get connected in further steps of the
dynamic programming. Again, see~\cite{ADF+11} for the details.

It is proved in~\cite{ADF+11} that rooted packings allow to carry out dynamic programming in order to determine whether an input graph $G$ contains a graph $F$ as a minor. It is easy to see that the number of distinct rooted packings at a separator $B$ is upper-bounded by $f(t,F):=2^{t \log t} \cdot r^t \cdot 2^{r^2}$, where $ t \geq |B| $. In particular, this proves that when $\mc{G}$ is the class of graphs excluding a fixed graph $H$ on $h$ vertices as a minor, then the index of the equivalence relation \eq\G is bounded by $2^{t \log t} \cdot h^t \cdot 2^{h^2}$.

\smallskip
\noindent\textbf{The encodings generator \C.}
Let $G \in \B$ with boundary $\partial (G)$ labeled with $\Lambda(G)$. The function \C maps $\Lambda(G)$ to a set  $\C(\Lambda(G))$ of encodings.
Each $R\in \C(\Lambda(G))$ is a set of at most $|\Lambda(G)|$ rooted packings $\{(\mc{A}_i, S_{F_i}^*,S_{F_i} ,\psi_i, \chi_i) \mid F_i \in\mc{F} \}$, where each such rooted packing encodes a potential model of a minor $F_i \in \mc{F}$ (multiple models of the same graph are allowed).

\smallskip
\noindent\textbf{The language \L.}
For a packing of models $\mc{S}$, we say that $(G,\mc{S},R)$ belongs to the language $\L$ (or that $\mc{S}$ is a packing of models \emph{satisfying} $R$) if there is a packing of potential models matching with the rooted packings of $R$ in $G \setminus \bigcup_{\Phi \in \mc{S}} \Phi(F)$.

Note that we allow the entirely realized models of $\mc{S}$ to intersect $\partial (G) $ arbitrarily, but they must not intersect potential models imposed by $R$.
\newline

As mentioned in the introduction, the natural definition of the maximization function does not provide a confined encoder, hence we need to use the relevant function \fbar. In order to define this function we note that, given a separator $B$ and a subgraph $G_B$, a (partial) solution naturally {\sl induces} an encoding $R_B \in \C(\Lambda(G_B))$, where the rooted packings correspond to the intersection of models with $B$.

Formally, let $G$ be a $t$-boundaried graph with boundary $A$ and let $\mc{S}$ be a partial solution satisfying some $R_A \in \C(\Lambda(G))$. Let also $\mc{P}$ be the set of potential models matching with the rooted packings in $R_A$. Given a separator $B$ in $G$, we define the induced encoding $R_B = \{(\mc{A}_i, S_{F_i}^*,S_{F_i} ,\psi_i, \chi_i) \mid \Phi_i \in \mc{S}\cup\mc{P}\} \in \C(\Lambda(G_B))$ such that
for each (potential) model $\Phi_i \in \mc{S} \cup \mc{P}$ of $F_i \in \mc{F}$ intersecting $B$,

\begin{itemize}
\item[$\bullet$] $\mc{A}_i$ contains elements of the form $B\cap C$, where $C$ is a connected component of the graph induced by  $V(\Phi_i(v)) \cap V(G_B)$, with $v \in V(F_i)$;
\item[$\bullet$] $\psi_i$ maps each element of $\mc{A}_i$ to its corresponding vertex in $F_i$; and
\item[$\bullet$] $S_{F_i}^*,S_{F_i}$, correspond to the vertices of $F_i$ whose vertex models intersect $G_B$ and $B$, respectively.
\end{itemize}

Clearly, the set of models of $\mc{S}$ entirely realized in $G_B$ is a partial solution satisfying $R_B$.

Provided with a formal definition of an induced encoding, and following the description given in Section~\ref{sec:applications}, we can state the definition of an irrelevant encoding for our problem.
 Let $G \in \B$ with boundary $A$ and let $R_A$ be an encoding. An encoding $R_A$ is \emph{irrelevant for \fbar} if there exists a certificate $\mc{S}$ such that $(G,\mc{S},R_A)\in \L$ and $|\mc{S}|= \f(G,R_A)$, and a separator $B \subseteq V(G)$ with $|B| \leq t$ and $B \neq A$, such that $\mc{S}$ {\sl induces} (as defined above) an encoding $R_B$ in the graph $G_B \in \B$ with $ \fbar(G_B,R_B) = -\infty$.

\smallskip
\noindent\textbf{The function \fbar.}
Let $G \in \B$ with boundary $A$ and let $g(t) =t$. We define
the function $\fbar$ as 
\begin{equation} \label{eq: relevant f-packing}
\fbar(G,R_A) =\
\left\{\begin{array}{lll}
  & -\infty,  & \text{if } \f(G,R_A) + g(t) < \\
  &           &      ~~~~~~~~~~~~~ \max \{\f(G,R): R \in \C(\Lambda(G)) \}   \\
  &           & \text{or if } R_A \text{ is irrelevant for } \fbar. \\
  & \f(G,R_A),  & \mbox{otherwise}.\\
\end{array}\right.
\end{equation}

In the above equation, \f is the natural maximization function associated with the encoder, that is, $\f(G,R)$ is the maximal number of (entire) models in $G$ which do not intersect potential models imposed by $R$. Formally,
\begin{equation*}\label{eq:fEmin}
\f(G,R) \ = \ \max \{k  : \exists \mc{S}, |\mc{S}| \geq k, (G,\mc{S},R) \in \L\}.
\end{equation*}

\smallskip
\noindent\textbf{The size of \E.}
Recall that $f(t,F):=2^{t \log t} \cdot r^t \cdot 2^{r^2}$ is the number of rooted packings for a minor $F$ of size $r$ on a boundary of size $t$.  If we let $r := \max_{F \in \mc{F}} |V(F)|$ and $J$ be any set of positive integers such that $\sum_{j\in J} j \leq t$, by definition of \E, it holds that
\begin{equation}\label{eq:sizeEncoderFPackSet}
s_{\E}(t) \ \leq \ (\sum_{j\in J} 2^{j \log j} \cdot r^j \cdot 2^{r^2}) \\
           \ \leq \ (\sum_{j\in J} 2^{t \log t} \cdot r^t) \cdot 2^{r^2} \\
           \ \leq \  t \cdot 2^{t \log t} \cdot r^t \cdot 2^{r^2} .
\end{equation}

Note that an encoding can also be seen as the rooted packing of the disjoint union of at most $t$ minors of $\mc{F}$.

\begin{fact} \label{fait: rp}
Let $G \in \B$ with boundary $A$, let $\Phi$ be a model (resp. a potential model matching with a rooted packing defined on $A$) of a graph $F$ in $G$, let $B$ be a separator of $G$, and let $G_B \in \B$ be as in Definition \ref{defi:DP-friend}. Let $(\mc{A}, S_{F}^*,S_{F} ,\phi, \chi)$ be the rooted packing induced by $\Phi$ (as defined above). Let $G_B' \in \B$ with boundary $B$ and let $G'$ be the graph obtained by replacing $G_B$ with $G_B'$.
If $G_B'$ has a potential model $\Phi_B'$ matching with $(\mc{A}, S_{F}^*,S_{F} ,\phi, \chi)$, then $G'$ has a model (resp. a potential model) of $F$.
\end{fact}

\begin{proof}
Let us build  a model (resp. a potential model) $\Phi'$ of $F$ in $G'$. For every vertex $v$ in $V(F) \setminus S_{F}^*$, 	 we set $\Phi'(v) = \Phi(v)$.
For every vertex $v$ in $S_{F}^* \setminus S_{F}$, 	we set $\Phi'(v) = \Phi_B'(v)$.
For every vertex $v$ in $S_{F}$, 					we set $\Phi'(v) = \Phi(v)[V(G) \setminus V(G_B)] \oplus \Phi_B'(v)$.
As $\Phi(v)$ is connected and the connected components in $\Phi_B'(v)$ have the same boundaries than the ones in $\Phi(v)[V(G_B)]$ (by definition of rooted packing), it follows that $\Phi'(v)$ is connected.
Note that $\Phi'(v)$ do not intersect $\Phi'(u)$, since $\Phi(v),\Phi_B'(v)$ do not intersect $\Phi'(u)$ for any $u \in V(F)$.

For every edge $e$ in $V(F) \times V(F) \setminus S_{F}^*$
     or such that $\chi(e) = 0$     				we set $\Phi'(e) = \Phi(e)$.
For every edge $e$ in $S_{F}^*  \times S_{F}^* \setminus S_{F}$
     or such that $\chi(e) = 1$ 					we set $\Phi'(e) = \Phi_B'(e)$.
Since $B$ is a separator in $G$, $S_{F}$ is a separator in $F$ and there is no edge in $V(F) \setminus S_{F}^* \times S_{F}^* \setminus S_{F}$. Since $\Phi, \Phi_B'$ are (potential) models, the edges $\Phi'(e), e \in E(F)$ are obviously distinct and if $e = \{u,v\}$, then $\Phi'(e)$ as one endpoint in $\Phi'(u)$ and the other in $\Phi'(v)$.
\end{proof}

\begin{figure}[h!b]
\begin{center}
\scalebox{.734}{\includegraphics{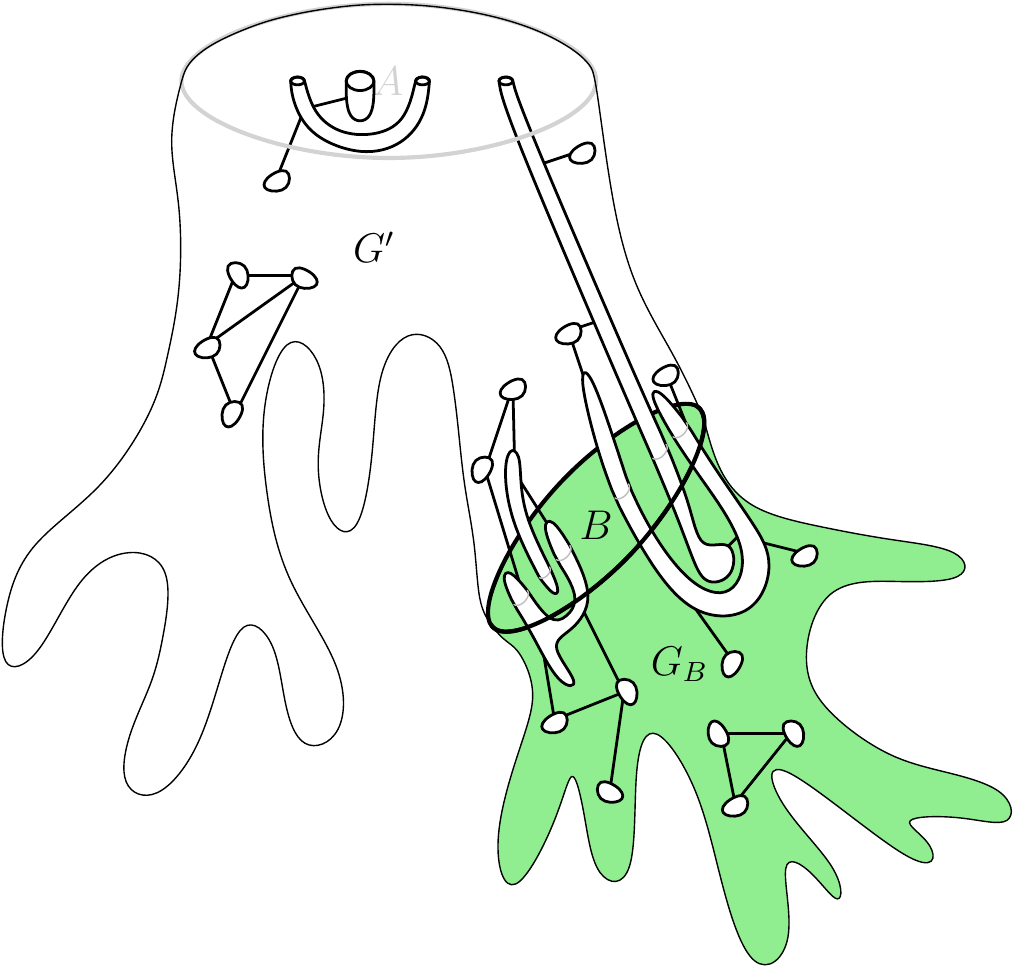}\includegraphics{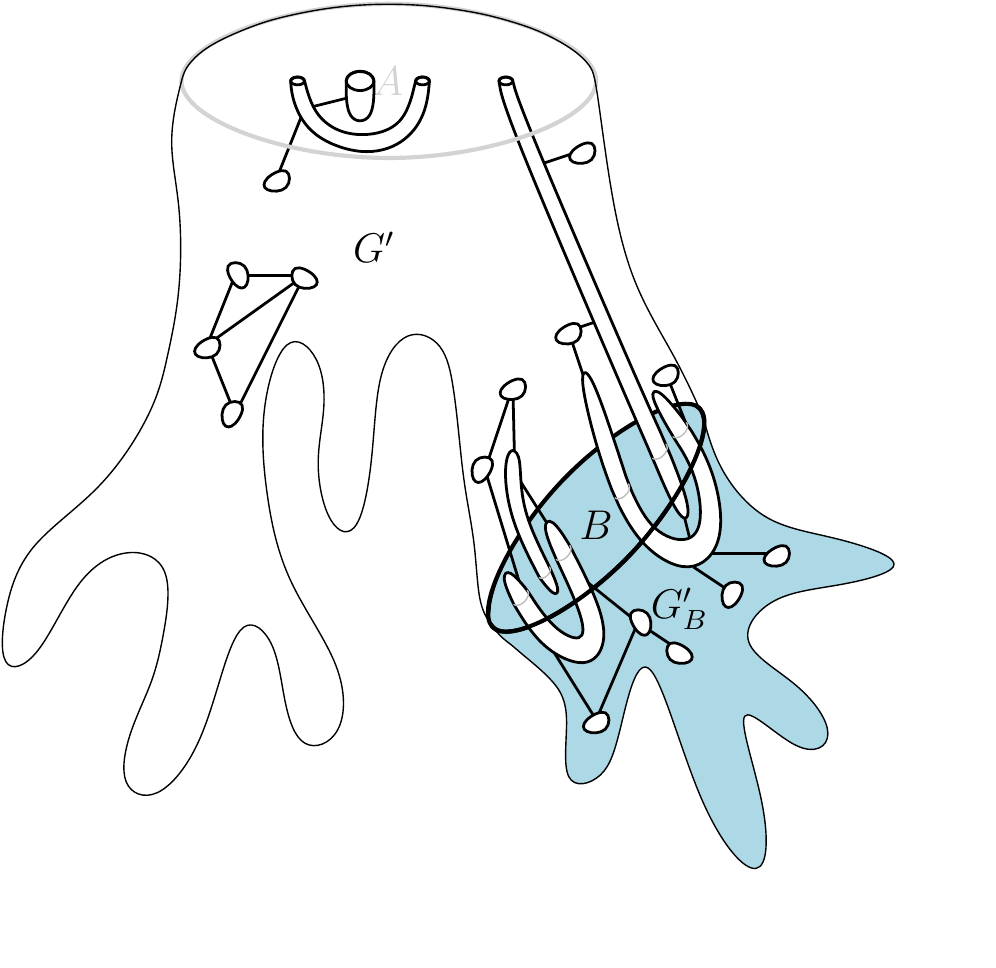}}
\end{center}
\caption{Illustration of a protrusion replacement  for \FPack.}
\label{fig:replacement-packing}
\end{figure}

See Figure~\ref{fig:replacement-packing} for an illustration of the scenario described in the statement of Fact~\ref{fait: rp}.

\begin{lemma} \label{lem: Pack DPfriend}
The encoder $\E$ is a $g$-confined $\textsc{c}\mc{F}\textsc{P}$-encoder for $g(t)=t$.\! Furthermore, if \G is an arbitrary class of graphs, then the equivalence relation \eqsG\E is DP-friendly.
\end{lemma}

\begin{proof} Let us first show that the encoder \E is a  $\textsc{c}\mc{F}\textsc{P}$-encoder. Indeed, if $G$ is a $0$-boundaried graph, then $\C(\emptyset)$ consists of a single encoding $R_{\emptyset}$ (an empty set of rooted packings), and by definition of \L, any $\mc{S}$ such that $(G,\mc{S},R_\emptyset) \in \L$ is a packing of models.
According to Equation~(\ref{eq: relevant f-packing}), there are two possible values for $\fbar(G,R_\emptyset)$: either $\f(G,R_\emptyset)$, which by definition equals $f^\Pi(G)$, or $-\infty$.
Let $\mc{S}$ be a packing of models of size $f^\Pi(G)$, and assume for contradiction that $\fbar(G,R_\emptyset) = - \infty$. Then, by a recursive argument we can assume that there is a separator $B$ of size at most $t$ and a subgraph $G_B$ of $G$ as in Definition \ref{defi:DP-friend}, such that $\mc{S}$ induces $R_B$ and $ \f(G_B,R_B) + t < \max \{\f(G,R): R \in \C(I) \}$.
Let $M$ be the set of models entirely realized in $G_B$. We have $|M| =\f(G_B,R_B)$, as otherwise $\mc{S}$ is not maximal. Let $M_B$ be the set of models intersecting $B$, so we have $|M_B| \leq t$. Finally, let $M_0$ be a packing of models in $G_B$ of size $\max \{\f(G,R), R \in \C(I) \}$.
Clearly, $\mc{S} \setminus (M \cup M_B) \cup M_0$ is a packing of models smaller than $\mc{S}$ (by optimality), that is, $ |M_0| \leq |M| + t$, a contradiction with the definition of \fbar. Hence $\fbar(G,R_\emptyset) = f^\Pi(G)$.

By definition of the function \fbar, the encoder \E is $g$-confined for $g: t \mapsto t$.


It remains to prove that the equivalence relation \eqsG\E is DP-friendly for $g(t)=t$.
Due to Fact \ref{fait:equiv}, it suffices to prove that \eqs\E is DP-friendly.
Let $G \in \B$ with boundary $A$, let $B$ be any separator of $G$, and let $G_B$ be as in Definition \ref{defi:DP-friend}.
The subgraph $G_B$ can be viewed as a $t$-boundaried graph with boundary $B$. We define $H \in \B$ to be the graph induced by $V(G) \setminus (V(G_B) \setminus B)$, with boundary $B$ (that is, we forget boundary $A$) labeled in the same way than $G_B$. Let $G_B' \in \B$ such that $G_B \eqs\E G_B'$ and let $G' = H \oplus G_B'$, with boundary $A$. 
We have to prove that $G \eqs\E G' $ and $\D\E(G,G') = \D\E(G_B,G_B')$, that is, that $\fbar(G,R_A) = \fbar(G',R_A) + \D\E(G_B,G_B')$ for all $R_A \in \C(\Lambda(G))$.

Let $R_A$ be an encoding defined on $A$.
Assume first that $\fbar(G,R_A) \neq - \infty$.
Let $\mc{S} = M \cup M_B \cup M_H$ be a packing of models satisfying $R_A$ with size $ \fbar(G,R_A) $ in $G$, with $M$ being the set of models entirely contained in $G_B$, $M_H$ the set of models entirely contained in $V(H)\setminus B$, and $M_B$ the set of models intersecting $B$ and $H$. Notice that $M, M_B , M_H$ is a partition of $\mc{S}$.
Let $\mc{P}$ be the set of potential models matching with the rooted packings in $R_A$.
Let also $R_B \in \C(\Lambda(G_B))$ be the encoding induced by $\mc{S} \cup \mc{P}$. 

Observe that $\fbar(G_B,R_B) \neq - \infty$, as otherwise, by definition of the relevant function $\fbar$, we would have that $\fbar(G,R_A) = - \infty$. Also,  by construction of $R_B$ it holds that $|M| = \fbar(G_B,R_B)$, as otherwise $\mc{S}$ would not be not maximum.
Let $M'$ be a packing of models of $\mc{F}$ in $G_B'$ such that $(G_B',M',R) \in \L$ and of maximum cardinality, that is, such that $|M'| = \fbar(G_B',R_B)$.
Consider now the potential models matching with $R_B$. There are two types of such potential models. The first ones match with rooted packings defined by the intersection of models in $\mc{S}$ and $B$; we glue them with the potential models defined by $H \cap M_B$ to construct $M'_B$.
The other ones match with rooted packings defined by the intersection of potential model in $\mc{P}$ and $B$; we glue them with the potential models defined by $H \cap \mc{P}$ to construct $\mc{P'}$.
Observe that $|M_B| = |M_B'|$. As $G_B \eqs\E G_B'$ and {$\fbar(G_B,R_B) \neq -\infty$}, we have that $|M'| =  \fbar(G_B,R_B) + \D\E(G_B,G_B')$, and therefore $|M' \cup M_B' \cup M_H| =  \fbar(G_B,R_B) + \D\E(G_B,G_B') + |M_B| + |M_H| = \fbar(G,R_A) + \D\E(G_B,G_B')$. 

By definition we have that $M_H$ and $M'$ are packings of models.
The set $M'_B$ contains vertex-disjoint models by Fact \ref{fait: rp}.
Note that
models in $M_H \cup M'$ are vertex-disjoint (because $V(H) \cap V(G_B) = \emptyset$),
models in $M_H \cup M'_B$ are vertex-disjoint (because the ones in $M_H \cup M_B$ are vertex-disjoint), and
models in $M' \cup M'_B$ are vertex-disjoint (because $M'$ satisfies $R_B$). 
Hence $M_H \cup M' \cup M'_B$ is a packing of models.

It remains to prove that $M_H \cup M' \cup M'_B$ satisfies $R_A$.
The set $\mc{P'}$  contains vertex-disjoint potential models by Fact \ref{fait: rp}.
Models in $\mc{P'} \cup M'$ are vertex-disjoint, as $M'_B$ satisfies $R_B$.
Models in $\mc{P'} \cup M'_B$ are vertex-disjoint by definition of $R_B$. Finally, models in $\mc{P'} \cup M_H$ are vertex-disjoint since $\mc{S}$ satisfies $R_A$.

It follows that $G'$ has a packing of models satisfying $R_A$ of size $ \fbar(G,R_A) + \D\E(G_B,G_B')$, that is,
$G \eqs\E G'$ and $\D\E(G,G') =  \D\E(G_B,G_B')$.

Assume now that $\fbar(G,R_A) = - \infty$.
If $\fbar(G',R_A) \neq - \infty$, then applying the same arguments as above we would have that $\fbar(G,R_A) \neq - \infty$, a contradiction. \end{proof}

\subsection{A linear kernel for \FPack} \label{ssec: FPack kernel}

We are now ready to provide a linear kernel for \textsc{Connected-Planar-$\mc{F}$-Packing}.

\begin{theorem} \label{theo: Kernel FPack}

Let $\mc{F}$ be a finite family of connected graphs containing at least one planar graph on $r$ vertices, let $H$ be an $h$-vertex graph, and let \G be the class of $H$-minor-free graphs. Then $\textsc{cp}\mc{F}\textsc{P}_\G$ admits a constructive linear kernel of size at most $f(r,h)\cdot k$, where $f$ is an explicit function depending only on $r$ and $h$, defined in Equation \eqref{eq: kernelFPacking}.
\end{theorem}

\begin{proof}
By Lemma \ref{lem: FPack prot decompo}, given an instance $(G,k)$  we can either conclude that $(G,k)$ is a \YES-instance of $\textsc{cp}\mc{F}\textsc{P}_\G$, or build in linear time an $((\alpha_{H} \cdot t )\cdot k', 2t + h)$-protrusion decomposition of $G$, where $\alpha_{H},t, k'$ are defined in the proof of Lemma \ref{lem: FPack prot decompo}.

We now consider the encoder \E defined in Subsection \ref{ssec: FPack encod}. By Lemma \ref{lem: Pack DPfriend}, \E is a $g$-confined $\textsc{cp}\mc{F}\textsc{P}_\G$-encoder and \eqsG\E is DP-friendly, where $g(t)=t$ and \G is the class of $H$-minor-free graphs. An upper bound on $s_{\E}(t)$ is given in Equation \eqref{eq:sizeEncoderFPackSet}. Therefore, we are in position to apply Corollary \ref{coro:main} and obtain a linear kernel for $\textsc{cp}\mc{F}\textsc{P}_\G$ of size at most
\begin{equation}\label{eq: kernelFPacking}
(\alpha_{H} \cdot t) \cdot (b\left(\E, g,t ,\mathcal{G}\right)+1) \cdot k' \ , \text{ where}
\end{equation}
\begin{itemize}
\item[$\bullet$] $b\left(\E, g, t ,\mathcal{G}\right)$ is the function defined in Lemma \ref{lem:progres size};
\item[$\bullet$] $t$ is the bound on the treewidth provided by Corollary \ref{coro: tw modul FDel};
\item[$\bullet$] $k' $ is the parameter of \textsc{$\mc{F}$-Deletion} provided by Theorem \ref{theo: Erdos Posa}; and
\item[$\bullet$] $\alpha_H $ is  the constant provided by Theorem \ref{theo:prot dec}.
\end{itemize}\vspace{-.65cm}
\end{proof}
%

By using the recent results of Chekuri and Chuzhoy \cite{ChCh13}, it can be shown that the factor $\alpha_H = O( h^2 2 ^{ O(h \log h) })$ in Theorem \ref{theo: Erdos Posa} can be replaced with $h^{O(1)}$. However, in this case this would not directly translate into an improvement of the size of the kernel given in Equation \eqref{eq: kernelFPacking}, as the term $h^{O(1)}$ would be dominated by the term $f_m (h) = 2^{ O(h^2 \log h)}$. 

\renewcommand{\Epi}[2]{\ensuremath{\mc{#1 E}_{\ell\!\mc{F}\!\sc{P}}^{#2}}\xspace}
\renewcommand{\Ebar}{{\Epi{\bar}{}}\xspace}
\renewcommand{\fbar}{\ensuremath{\bar f^{\Epi{}{}}_g}\xspace}
\newcommand{\Scat}{\textsc{Scattered Set}\xspace}

\section{Application to \rFPack} \label{sec: rFPack}

We now consider the scattered version of the packing problem.
Given a finite set of graphs  $\mc{F}$ and a positive integer $\ell$, the \textsc{$\ell$-$\mc{F}$-Packing} problem is defined as follows.

\vspace{.4cm}
\begin{boxedminipage}{.9\textwidth}
\textsc{$\ell$-$\mc{F}$-Packing}
\vspace{.1cm}

\begin{tabular}{ r l }
\textbf{Instance:}  & A graph $G$ and a non-negative integer $k$. \\
\textbf{Parameter:} & The integer $k$.\\
\textbf{Question:}  & Does $G$ have $k$ subgraphs $G_1,\ldots,G_k$ pairwise at distance at \\
                    & ~~~least $\ell$, each containing some graph from $\mc{F}$ as a minor?\\
\end{tabular}
\end{boxedminipage}
\vspace{.4cm}

We again consider the version of the problem where all the graphs in $\mc{F}$ are connected and at least one is planar, called \textsc{Connected-Planar-$\ell$-$\mc{F}$-Packing} ($\textsc{cp}\ell\mc{F}\textsc{P}$).

We obtain a linear kernel for \textsc{Connected-Planar \rFPack} on the family of graphs excluding a fixed apex graph $H$ as a minor. We use again the notions of model, packing of models, and rooted packing.


\subsection{A protrusion decomposition for an instance of \rFPack}\label{ssec: rFPack prot decompo}

In order to obtain a linear protrusion decomposition for \rFPack, a natural idea could be to prove an Erd\H{o}s-P\'osa property at distance $\ell$, generalizing the approach for \textsc{$\mc{F}$-Packing} described in Section~\ref{sec: FPack}. Unfortunately, the best known Erd\H{o}s-P\'osa relation between a maximum $\ell$-$\cal F$-packing and a minimum $\ell$-$\cal F$-deletion set is not linear. Indeed, by following and extending the ideas of Giannopoulou \cite[Theorem 8.7 in Section 8.4]{Archontia-PhD} for the special case of cycles, it is possible to derive a bound of $O(k\sqrt{k})$, which is superlinear, and therefore not enough for our purposes. Proving a linear bound for this Erd\H{o}s-P\'osa relation, or finding a counterexample, is an exciting topic for further research.


We will use another trick to obtain the decomposition: we will (artificially) consider the \rFPack problem as a vertex-certifiable problem. Hence we propose the  formulation described below, which is clearly equivalent to the previous one. Using such a formulation, a natural question is whether  the \rFPack problem can fit into the framework for vertex-certifiable problems~\cite{KviaDP}. However, finding an appropriate encoder for this formulation does not seem an easy task, and it is more  convenient to
describe the encoder for \rFPack using the new framework designed for packing problems.

\vspace{.4cm}
\begin{boxedminipage}{.94\textwidth}
\textsc{$\ell$-$\mc{F}$-Packing}
\vspace{.1cm}

\begin{tabular}{ r l }
\textbf{Instance:}  & A graph $G$ and a non-negative integer $k$. \\
\textbf{Parameter:} & The integer $k$.\\
\textbf{Question:}  & Does $G$ have a set $\{v_1, \dots, v_k \}$ of $k$ vertices such that every $v_i$ \\
                    &  ~~~belongs to a subgraph $G_i$ of $G$ with $G_1,\ldots,G_k$ pairwise at\\
                    & ~~~distance at least $\ell$ and each containing some graph from $\mc{F}$\\
                    &  ~~~as a minor?\\
\end{tabular}
\end{boxedminipage}
\vspace{.4cm}

With such a formulation, we are in position to use some powerful results from Bidimensionality theory. It is not so difficult to see that the \rFPack problem is
\emph{contraction-bidimensional}~\cite{FLST10}.
Then we can use Theorem~\ref{theo: tw mod} and obtain the following corollary.
Again, the bound on the treewidth is derived from the proof of Theorem \ref{theo: tw mod} in~\cite{FLST10}.

\begin{corollary} \label{coro: tw modul rFpack}
Let $\mc{F}$ be a finite set of graphs containing at least one $r$-vertex planar graph $F$, let $H$ be an $h$-vertex apex graph, and let \G be the class of $H$-minor-free graphs.
If $(G,k) \in \textsc{p}\ell\mc{F}\textsc{P}_\G$, then there exists a set $X\subseteq V(G)$
  such that $|X|= k$
  and $\tw(G-X)= O ((2r+\ell)^{3/2} \cdot \tau_H^3 \cdot f_c(h)^3 ) $. Moreover, given an instance  $(G,k)$ with $|V(G)|=n$, there is an algorithm running in time $O(n^{3})$ that either finds such a set $X$ or correctly reports that $(G,k) \notin \textsc{p}\ell\mc{F}\textsc{P}_\G$.

\end{corollary}

We are now able to construct a linear protrusion decomposition.

\begin{lemma} \label{lem: rFPack prot decompo}
Let $\mc{F}$ be a finite set of graphs containing at least one $r$-vertex planar graph $F$, let $H$ be an $h$-vertex apex graph, and let \G be the class of $H$-minor-free graphs.
Let $(G,k)$ be an instance of \textsc{Connected Planar}-\rFPack. If $(G,k)  \in \textsc{cp}\ell\mc{F}\textsc{P}_\G$, then we can construct  in polynomial time a linear protrusion decomposition  of $G$.
\end{lemma}

\begin{proof}
Given an instance $(G,k)$ of $\textsc{cp}\ell\mc{F}\textsc{P}_\G$, we run the algorithm given by Corollary \ref{coro: tw modul rFpack}.
If the algorithm is not able to find a treewidth-modulator $X$ of size $|X|= k$, then we can conclude that $(G,k) \notin \textsc{cp}\ell\mc{F}\textsc{P}_\G$.
Otherwise, we use the set $X$ as input to the algorithm given by Theorem~\ref{theo:prot dec}, which outputs in linear time an
$((\alpha_{H} \cdot t)\cdot k, 2t + h)$-protrusion decomposition of $G$, where
\begin{itemize}
\item[$\bullet$] $t =O ( (r+\ell)^{3/2} \cdot \tau_H^3 \cdot f_c(h)^3 )$ is provided by Corollary \ref{coro: tw modul rFpack}; and
\item[$\bullet$] $\alpha_H = O( h^2 2 ^{ O(h \log h) })$ is  the constant provided by Theorem \ref{theo:prot dec}.
\end{itemize}

This is an $\left( h^2 \cdot 2 ^{ O(h \log h) } \cdot  (r+\ell)^{3/2} \cdot \tau_H^3 \cdot f_c(h)^3 \cdot k,\ O ( (r+\ell)^{3/2} \cdot \tau_H^3 \cdot f_c(h)^3 )\right)$-pro\-trusion decomposition of $G$.
\end{proof}

\subsection{An encoder for \rFPack} \label{ssec: rFPack encod}

Our encoder \E for \rFPack is a combination of the encoder for \FPack and the one for $\ell$-\Scat that we defined in~\cite{KviaDP}.

\smallskip
\noindent\textbf{The encodings generator \C.}
Let $G \in \B$ with boundary $\partial (G)$ labeled with $\Lambda(G)$. The function \C maps $\Lambda(G)$ to a set  $\C(\Lambda(G))$ of encodings.
Each $R\in \C(\Lambda(G))$ is a pair $(R_P,R_S)$, where
\begin{itemize}
\item[$\bullet$] $R_P$ is a set of at most $|\Lambda(G)|$ rooted packings $\{(\mc{A}_i, S_{F_i}^*,S_{F_i} ,\phi_i, \chi_i) \mid i \in \Lambda(G), F_i \in\mc{F} \}$, where each such rooted packing encodes a potential model of a minor $F_i \in \mc{F}$ (that is, $R_P$ is an encoding of \FPack); and
\item[$\bullet$] $R_S$  maps label $j \in \Lambda(G)$ to an $|\Lambda(G)|$-tuple $(d,d_i,~ i \in \Lambda(G), i \neq j) \in [0,\ell+1]^{|\Lambda(G)|}$ (that is,  $R_S$ is an encoding of $\ell$-\Scat), for simplicity, since each label in $\Lambda(G)$ is uniquely associated with a vertex in $\partial(G)$, we denote by $R(v)$ the vector assigned by $R_S$ to label $\lambda(v)$.
\end{itemize}

\noindent\textbf{The language \L.}
For a packing of models $\mc{S}$, we say that $(G,\mc{S},R)$ belongs to the language $\L$ (or that $\mc{S}$ is a packing of models \emph{satisfying} $R$) if
\begin{itemize}
\item[$\bullet$] the models are pairwise at distance at least $\ell$,
      that is, for each $\Phi_1, \Phi_2 \in \mc{S}$  models of $F_1,F_2 \in \mc{F}$, respectively, $d_G(V(\Phi_1(F_1)),V(\Phi_2(F_2))) \geq \ell$;
\item[$\bullet$] there is a packing of potential models matching with the rooted packings of $R_P$ pairwise at distance at least $\ell$ and at distance at least $\ell$  from $\bigcup_{\Phi \in \mc{S}} \Phi(F)$; and
\item[$\bullet$] for any vertex $v \in \partial(G)$, if $(d,d_i) = R(v)$ then $d_G(v,\mc{S} \cup \mc{P}) \geq d$, and   $d_G(v,w) \geq d_{\lambda(w)}$, for any $w\in \partial(G)$.

\end{itemize}

Similarly to \FPack, we need the relevant version of the function \fbar.
Let $G\in \B$ with boundary $A$ and let $\mc{S}$ be a partial solution satisfying some $R_A \in \C(\Lambda(G))$. Let also $\mc{P}$ be the set of potential models matching with the rooted packings in $R_A$. Given a separator $B$ in $G$, and $G_B$ as in Definition \ref{defi:DP-friend}, we define the induced encoding $R_B = (R_P,R_S)$ as follows:
\begin{itemize}
\item[$\bullet$] $R_P$ is defined by the intersection of $B$ with models in $ \mc{S} \cup \mc{P}$, (as for \FPack); and
\item[$\bullet$] $R_S$ maps each $v \in B$ to $R(v) = (d_{G_B}(v,\mc{S} \cup \mc{P}), d_{G_B}(v,w), ~ w \in B)$.
\end{itemize}

The set of models of $\mc{S}$ entirely realized in $G_B$ is a partial solution satisfying $R_B$.

The definition of an irrelevant encoding is as described in Section ~\ref{sec:applications}.

\smallskip
\noindent\textbf{The function \fbar.}
Let $G \in \B$ with boundary $A$ and let $g(t) =2 t$.  We define
$\fbar$ as
\begin{equation} \label{eq: relevant f}
\fbar(G,R_A) =\
\left\{\begin{array}{lll}
  & -\infty,  & \text{if } \f(G,R_A) + 2t <  \\
  &           &  ~~~~~~\max \{\f(G,R): R \in \C(\Lambda(G)) \},\\
  &           & \text{or if } R_A \text{ is irrelevant for } \fbar. \\
  & \f(G,R_A),  & \mbox{otherwise}.\\
\end{array}\right.
\end{equation}
In the above equation,  \f is the natural optimization function defined as
\begin{equation}\label{eq:fEmin}
\f(G,R) \ = \ \max \{k \ :  \exists \mc{S}, |\mc{S}| \geq k, (G,\mc{S},R) \in \L\}.
\end{equation}

\smallskip
\noindent\textbf{Size of \E.}
Since $\C(I) = { \mc{C}^{ \mc{E}_{\!\mc{F}\!\sc{P}} } }
 \times ([0,\ell+1]^t)^t$, it holds that
\begin{equation}\label{eq:sizeEncoderrFPackSet}
s_{\E}(t) \ \leq \ s_{\bar{\mc{E}}_{\!\mc{F}\!\sc{P}}}(t) \times (\ell+2) ^ {t^2}.
\end{equation}

\begin{lemma} \label{lem: rPack DPfriend}
The encoder $\E$ is a $g$-confined $\textsc{c}\ell\mc{F}\textsc{P}$-encoder for $g(t)=2t$. Furthermore, if \G is an arbitrary class of graphs, then the equivalence relation \eqsG\E is DP-friendly.
\end{lemma}

\begin{proof}
We first prove that \E is a $\textsc{c}\ell\mc{F}\textsc{P}$-encoder. Obviously, $\{(G,\mc{S}) : (G,\mc{S},R_\emptyset) \in \L, R_\emptyset \in \C(\emptyset) \} = L^\Pi$. As in the proof of Lemma~\ref{lem: Pack DPfriend}, in order to show that $ \f(G,R_\emptyset) \neq -\infty$ we prove that the value computed by \fbar has not been truncated. Let $G, G_B$ and $\mc{S}, M, M_H, M_B, M_0$ as in proof of Lemma \ref{lem: Pack DPfriend}, and let $M_0^* = M_0 \setminus \{\Phi(F) : \Phi(F) \cap N_{r/2}(B) \neq \emptyset, F \in \mc{F} \}$ and $M_H^* = M_H \setminus \{\Phi(F) : \Phi(F) \cap N_{r/2}(B) \neq \emptyset, F \in \mc{F} \}$. $M_0^* \cup M_H^*$ is a scattered packing of size at least $|\mc{S}|-2t$.

The encoder \E is $g$-confined for $g: t \mapsto 2t$ by definition of \fbar.

Following the proof of Lemma \ref{lem: Pack DPfriend} again, let $G,G' \in \B$ with boundary $A$ and let $G_B,G_B',H \in \B$ with boundary $B$.
We have to prove that 
$\fbar(G,R_A) = \fbar(G',R_A) + \D\E(G_B,G_B')$ for every $R_A \in \C(\Lambda(G))$.

Let $R_A$ be an encoding defined on $A$.
Assume that $\fbar(G,R_A) \neq - \infty$.
Let $\mc{S} = M \cup M_B \cup M_H$ be a packing of models satisfying $R_A$ with size $ \fbar(G,R_A) $ in $G$, with $M,M_B,M_H$ as in the proof of Lemma \ref{lem: Pack DPfriend}. Let also $\mc{P}$ be the set of potential models matching with $R_A$ and let $R_B \in \C(\Lambda(G_B))$ be the encoding induced by $\mc{S} \cup \mc{P}$.

Observe that, by definition, $\fbar(G_B,R_B) \neq - \infty$. Hence there is a packing $M'$ in $G_B'$ of maximum cardinality and such that $(G_B',M',R) \in \L$. As in the proof of Lemma \ref{lem: Pack DPfriend}, we can define $M_B'$ to be the set of models obtained
from the potential models defined by the intersection of models in $M_B$ with $H$, glued to the ones in $G_B'$ matching with $R_B$.
We can also define $\mc{P'}$ to be the set of potential models obtained
from the potential models defined by the intersection of models in $M_B$ with $H$, glued to the ones in $G_B'$ matching with $R_B$.
As $G_B \eqs\E G_B'$ and following the argumentation in Lemma \ref{lem: Pack DPfriend} we have that $|M' \cup M_B' \cup M_H| = \fbar(G,R_A) + \D\E(G_B,G_B')$.

We already have that $\mc{S'} = M_H \cup M' \cup M'_B$ is a packing of models according to the proof of Lemma \ref{lem: Pack DPfriend}. It remains to prove that (potential) models in $\mc{S'} \cup \mc{P'}$ are pairwise at distance at least $\ell$. We follow the proof of \cite[Lemma 6]{KviaDP}. Let $P$ be a shortest path between any two models in $\mc{S'} \cup \mc{P'}$. We subdivide $P$ into maximal subpaths in $G_B'$ and maximal subpaths in $H$.
Clearly the length of a subpath in $H$ does not change.
Moreover, note that the length of a subpath in $G_B'$ with extremities $v,w \in B$ is at least $d_{G_B}(v,w)$, by definition of $R_B$. Note also
that the length of a subpath in $G_B'$ with an extremity in a model and the other $v \in B$ is at least $d_{G_B}(v,\mc{S})$, also by definition of $R_B$.
Therefore, the distance between any two models is indeed at least $\ell$.

It follows that $G'$ has a scattered packing of models satisfying $R_A$ of size $ \fbar(G,R_A) + \D\E(G_B,G_B')$, that is,
$G \eqs\E G'$ and $\D\E(G,G') =  \D\E(G_B,G_B')$. The case where $\fbar(G,R_A) = - \infty$ is easily handled as in Lemma \ref{lem: Pack DPfriend}.
\end{proof}

\subsection{A linear kernel for \rFPack} \label{ssec: FPack kernel}

We are now ready to provide a linear kernel for \textsc{Connected-Planar-$\ell$-$\mc{F}$-Packing}.

\begin{theorem} \label{theo: Kernel FPack}
Let $\mc{F}$ be a finite family of connected graphs containing at least one planar graph on $r$ vertices, let $H$ be an $h$-vertex apex graph, and let \G be the class of $H$-minor-free graphs. Then $\textsc{cp}\ell\mc{F}\textsc{P}_\G$ admits a constructive linear kernel of size at most $f(r,h,\ell)\cdot k$, where $f$ is an explicit function depending only on $r$, $h$, and $\ell$, defined in Equation \eqref{eq: kernel rFPack}.
\end{theorem}

\begin{proof}
By Lemma \ref{lem: rFPack prot decompo}, given an instance $(G,k)$  we can either report that $(G,k)$ is a \YES-instance of $\textsc{cp}\ell\mc{F}\textsc{P}_\G$, or build in linear time an $((\alpha_{H} \cdot t )\cdot k, 2t + h)$-protrusion decomposition of $G$, where $\alpha_{H}$ and $t$ are defined in the proof of Lemma \ref{lem: rFPack prot decompo}.

We now consider the encoder \E defined in Subsection \ref{ssec: rFPack encod}. By Lemma \ref{lem: rPack DPfriend}, \E is a $g$-confined $\textsc{cp}\ell\mc{F}\textsc{P}_\G$-encoder and \eqsG\E is DP-friendly, where $g(t)=2t$ and \G is the class of $H$-minor-free graphs. An upper bound on $s_{\E}(t)$ is given in Equation \eqref{eq:sizeEncoderrFPackSet}. Therefore, we are in position to apply Corollary \ref{coro:main} and obtain a linear kernel for $\textsc{cp}\ell\mc{F}\textsc{P}_\G$ of size at most
\begin{equation}\label{eq: kernel rFPack}
(\alpha_{H} \cdot t) \cdot (b\left(\E, g,t ,\mathcal{G}\right)+1) \cdot k' \ , \text{ where}
\end{equation}
\begin{itemize}
\item[$\bullet$] $b\left(\E, g, t ,\mathcal{G}\right)$ is the function defined in Lemma \ref{lem:progres size};
\item[$\bullet$] $t$ is the bound on the treewidth provided by Corollary \ref{coro: tw modul rFpack}; and
\item[$\bullet$] $\alpha_H $ is  the constant provided by Theorem \ref{theo:prot dec}.
\end{itemize}\vspace{-.7cm}
\end{proof}

%


\newcommand{\FMemb}{\textsc{$\cal F$-Packing with $\ell$-Membership}\xspace}
\renewcommand{\Epi}[2]{\ensuremath{\mc{#1 E}_{\!\mc{F}\!\sc{P}\!\ell\!\sc{M}}^{#2}}\xspace}
\renewcommand{\Ebar}{{\Epi{\bar}{}}\xspace}
\renewcommand{\fbar}{\ensuremath{\bar f^{\Epi{}{}}_g}\xspace}

\section{Application to \FMemb} \label{sec: FMemb}

Now we consider a generalization of the \FPack problem that allows models to be {\sl close} to each other (conversely to \rFPack, which asks for scattered models). That is, we consider the version for minors of the \textsc{$\mc{F}$-Subgraph-Packing with $\ell$-Membership} defined in~\cite{FernauLR15}. Let $\mc{F}$ be a finite set of graphs. For every integer $\ell \geq 1$, we define the \FMemb  problem as follows.

\vspace{.4cm}
\begin{boxedminipage}{.9\textwidth}
\textsc{\FMemb}
\vspace{.1cm}

\begin{tabular}{ r l }
\textbf{Instance:}  & A graph $G$ and a non-negative integer $k$. \\
\textbf{Parameter:} & The integer $k$.\\
\textbf{Question:}  & Does $G$ have $k$ subgraphs $ G_1,\dots,G_k$ such that\\
                    & ~~~each subgraph contains some graph from $\cal F$ as a minor,\\
                    & ~~~and each vertex of $G$ belongs to at most $\ell$ subgraphs?\\
\end{tabular}
\end{boxedminipage}
\vspace{.4cm}

 We again consider the version of the problem where all the graphs in $\mc{F}$ are connected and at least one is planar, called \textsc{Connected-Planar-\FMemb} ($\textsc{cp}\mc{F}\textsc{P}\ell\textsc{M}$).

We obtain a linear kernel for \textsc{Connected-Planar-\FMemb} on the family of graphs excluding a fixed graph $H$ as a minor. We use again the notions of model, packing of models, and rooted packing.

\smallskip

Now, for an arbitrary graph, a certificate for \FMemb is a \emph{packing of models with $\ell$-membership}, defined as follows.

\begin{definition} \label{defi: packing models}
Given a set $\mc{F}$ of minors and a graph $G$, a \emph{packing of models with $\ell$-membership} $\mc{S}$ is a set of models such that each vertex of $G$ belongs to at most $\ell$ models, that is, to at most $\ell$ subgraphs $\Phi(F)$ for $\Phi \in \mc{S}, F \in  \mc{F}$.
\end{definition}

Note that the above definition is equivalent to saying that each vertex of $G$ belongs to at most $\ell$ vertex-models, since vertex-models of a model are vertex-disjoint.

\subsection{A protrusion decomposition for an instance of \FMemb}\label{ssec: FMemb prot decompo}

In order to find a linear protrusion decomposition, we use again the \emph{Erd\H{o}s-P\'{o}sa property}, as we did in Subsection \ref{ssec: FPack prot decompo}.
The construction of a linear protrusion decomposition becomes straightforward from the fact that a packing of models is in particular a packing of models with $\ell$-membership for every integer $\ell \geq 1$.

\begin{lemma} \label{lem: FMemb prot decompo}
Let $\mc{F}$ be a finite set of graphs containing at least one $r$-vertex planar graph $F$, let $H$ be an $h$-vertex graph, and let \G be the class of $H$-minor-free graphs.
Let $(G,k)$ be an instance of $\textsc{cp}\mc{F}\textsc{P}\ell\textsc{M}_\G$. If $(G,k)  \notin \textsc{cp}\mc{F}\textsc{P}\ell\textsc{M}_\G$, then we can construct in polynomial time a linear protrusion decomposition of $G$.
\end{lemma}

\begin{proof}
It suffices to note that if $\mc{S}$ is a packing of models of size $k$, then it is in particular a packing of models with $\ell$-membership for every integer $\ell \geq 1$. Hence, if $(G,k)  \notin \textsc{cp}\mc{F}\textsc{P}r\textsc{M}_\G$ then $(G,k)  \notin \textsc{cp}\mc{F}\textsc{P}_\G$ and we can apply Lemma \ref{lem: FPack prot decompo}.
\end{proof}

\subsection{An encoder for \FMemb} \label{ssec: FMemb encod}

Our encoder \E for \FMemb uses again the notion of rooted packing, but now we allow the rooted packings to intersect.

\smallskip
\noindent\textbf{The encodings generator \C.}
Let $G \in \B$ with boundary $\partial (G)$ labeled with $\Lambda(G)$. The function \C maps $\Lambda(G)$ to a set  $\C(\Lambda(G))$ of encodings.
Each $R\in \C(\Lambda(G))$ is a set of at most $\ell \cdot|\Lambda(G)|$ rooted packings$\{(\mc{A}_i, S_{F_i}^*,S_{F_i} ,\phi_i, \chi_i) \mid F_i \in\mc{F} \}$, where each such rooted packing encodes a potential model of a minor $F_i \in \mc{F}$ (multiple models of the same graph are allowed).

\smallskip
\noindent\textbf{The language \L.}
For a packing of models with $\ell$-membership $\mc{S}$, we say that $(G,\mc{S},R)$ belongs to the language $\L$ (or that $\mc{S}$ is a packing of models with $\ell$-membership \emph{satisfying} $R$) if there is a packing of potential models with $\ell$-membership matching with the rooted packings of $R$ in $G \setminus \{u : u \in \Phi_1(F_1), \dots, u \in \Phi_\ell(F_\ell); \Phi_i \in \mc{S}, F_i \in \mc{F}  \}$, that is, such that each vertex belongs to at most $\ell$ models or potential models.

\smallskip
\noindent\textbf{The function \fbar.}
Similarly to \FPack, we need the relevant version of the function \fbar.
The function  \fbar is defined exactly as the one for \FPack in Section \ref{sec: FPack} (in particular, the encoding induced by a partial solution is also the set of rooted packings defined by the intersection of the partial solution and the separator).

\smallskip
\noindent\textbf{The size of \E.}
Note that the encoder contains at most $\ell t$ rooted packings  on a boundary of size $t$.
Hence, if we let $r := \max_{F \in \mc{F}} |V(F)|$, and $J$ be any set such that $\sum_{j\in J } j\leq \ell  t$ and $\forall j\in J, j \leq t$,  by definition of \E it holds that 
\begin{equation*}
s_{\E}(t) 
            \ \leq \  \ell t \cdot 2^{t\log t} \cdot r^t \cdot 2^{r^2}.
\end{equation*}
It just remains to prove that the relation \eqsG\E is DP-friendly. Note that in the encoder, the only difference with respect to \FPack is that rooted packings are now allowed to intersect. Namely, the constraint on the intersection is that each vertex belongs to at most $\ell$ models. This constraint can easily be verify locally, so no information has to be transmitted through the separator. Hence, the proof of the following lemma is exactly the same as the proof of Lemma \ref{lem: Pack DPfriend}, and we omit it.

\begin{lemma} \label{lem: Memb DPfriend}
The encoder \E is a $g$-confined $\textsc{c}\mc{F}\textsc{P}\ell\textsc{M}$-encoder for $g(t)=t$. Furthermore, if \G is an arbitrary class of graphs, then the equivalence relation \eqsG\E is DP-friendly.
\end{lemma}

\subsection{A linear kernel for \FMemb} \label{ssec: FPack kernel}

We are now ready to provide a linear kernel for \textsc{Connected-Planar} \FMemb.

\begin{theorem} \label{theo: Kernel FMemb}
Let $\mc{F}$ be a finite family of connected graphs containing at least one planar graph on $r$ vertices, let $H$ be an $h$-vertex graph, and let \G be the class of $H$-minor-free graphs. Then $\textsc{cp}\mc{F}\textsc{P}\ell\textsc{M}$ admits a constructive linear kernel of size at most $f(r,h,\ell)\cdot k$, where $f$ is an explicit function depending only on $r$, $h$, and $\ell$. 
\end{theorem}

The proof of the above theorem is exactly the same as the one of Theorem \ref{theo: Kernel FPack}, the only difference being in the size $s_{\E}(t)$ of the encoder, and hence in the value of $b\left(\E, g,t ,\mathcal{G}\right)$.

\newcommand{\FSub}{\ensuremath{\textsc{$\mc{F}$-Subgraph-Packing}}\xspace}
\newcommand{\lFSub}{\ensuremath{\textsc{$\ell$-$\mc{F}$-Subgraph-Packing}}\xspace}
\newcommand{\FSubM}{\ensuremath{\textsc{$\mc{F}$-Subgraph-Packing with $\ell$-Membership}}\xspace}
\renewcommand{\Epi}[2]{\ensuremath{\mc{#1 E}_{\!\mc{F}\!\sc{S\!P}}^{#2}}\xspace}
\newcommand{\El}      {\ensuremath{\mc{E}_{\!\ell\mc{F}\!\sc{S\!P}}}\xspace}
\newcommand{\Em}      {\ensuremath{\mc{E}_{\!\mc{F}\!\sc{S\!P}\ell\sc{M}}}\xspace}
\renewcommand{\Ebar}{{\Epi{\bar}{}}\xspace}
\renewcommand{\fbar}{\ensuremath{\bar f^{\Epi{}{}}_g}\xspace}

\newcommand{\A}{\ensuremath{A}\xspace}

\section{Application to \FSub} \label{sec: FSub}

In this section we apply our framework to problems where to objective is to pack subgraphs. The \FSub problem consists in finding vertex-disjoint subgraphs (instead of minors) isomorphic to graphs in a given finite family $\cal F$. Similarly to \textsc{$\mc{F}$-(Minor)-Packing}, we study two more generalizations of the problem, namely  the \lFSub, asking for subgraphs at distance $\ell$ from each other, and the \FSubM problem \cite{FernauLR15} that allows vertices to belong to at most $\ell$ subgraphs. 
Let $\mc{F}$ be a finite set of graphs and let $\ell \geq 1$ be an integer. The \FSub, the \lFSub, and the \FSubM problems are defined as follows.

\vspace{.4cm}
\begin{boxedminipage}{.93\textwidth}
\textsc{\FSub}
\vspace{.1cm}

\begin{tabular}{ r l }
\textbf{Instance:}  & A graph $G$ and a non-negative integer $k$. \\
\textbf{Parameter:} & The integer $k$. \\
\textbf{Question:}  & Does $G$ have $k$ vertex-disjoint subgraphs \\ & ~~~$G_1,\ldots,G_k$, each isomorphic to a graph in $\mc{F}$? \\
\end{tabular}
\end{boxedminipage}
\vspace{.1cm}

\vspace{.4cm}
\begin{boxedminipage}{.93\textwidth}
\textsc{\lFSub}
\vspace{.1cm}

\begin{tabular}{ r l }
\textbf{Instance:}  & A graph $G$ and two non-negative integers $k$ and $\ell$. \\
\textbf{Parameter:} & The integer $k$. \\
\textbf{Question:}  & Does $G$ have $k$ subgraphs $G_1,\ldots,G_k$ pairwise at distance\\
&  ~~~at least $\ell$ and  each isomorphic to a graph in $\mc{F}$? \\
\end{tabular}
\end{boxedminipage}
\vspace{.1cm}

\vspace{.4cm}
\begin{boxedminipage}{.93\textwidth}
\textsc{\FSubM}
\vspace{.1cm}

\begin{tabular}{ r l }
\textbf{Instance:}  & A graph $G$ and two non-negative integers $k$ and $\ell$. \\
\textbf{Parameter:} & The integer $k$. \\
\textbf{Question:}  & Does $G$ have $k$ subgraphs $G_1,\ldots,G_k$,  each isomorphic to \\ 
& ~~~in $\mc{F}$, and a graph such that each vertex of $G$ belongs  \\
& ~~~to at most $\ell$ subgraphs? \\
\end{tabular}
\end{boxedminipage}
\vspace{.4cm}

Again, for technical reasons, we consider the versions of the above problems where all the graphs in $\mc{F}$ are connected, called \textsc{Connected} \FSub ($\textsc{c}\mc{F}\textsc{SP}$), \textsc{Connected}\! \lFSub ($\textsc{c}\ell\mc{F}\textsc{SP}$),\! and \textsc{Connected}
\textsc{$\mc{F}$-Subgraph-Pack\-ing} \textsc{with $\ell$-Membership}
($\textsc{c}\mc{F}\textsc{SP$\ell$M}$), respectively. As in Section  \ref{sec: FPack}, connectivity is necessary to use the equivalent notion of rooted packings. Furthermore, in this section we also need connectivity to build the protrusion decomposition, whereas the presence of a planar graph in $\mc{F}$ is not mandatory anymore.

Similarly to \FPack, we establish a relation between instances of ${\cal F}\mbox{\sc -Subgraph}$
$\mbox{\sc-Packing}$
(and its variants) and instances of \textsc{$d$-Dominating Set} for an appropriate value of $d$. Therefore we also define this problem. Note that here we do not use any Erd\H{o}s-P\'osa property to establish this relation.

\vspace{.4cm}
\begin{boxedminipage}{.93\textwidth}
\textsc{\textsc{$d$-Dominating Set}}
\vspace{.1cm}

\begin{tabular}{ r l }
\textbf{Instance:}  & A graph $G$ and two non-negative integers $k$ and $d$. \\
\textbf{Parameter:} & The integer $k$. \\
\textbf{Question:}  & Is there a set $D$ of vertices in $G$ with size at most $k$,\\
                    & ~~~such that for every vertex $v \in V(G)$, $N_d[v] \cap D \neq \emptyset$?
\end{tabular}
\end{boxedminipage}
\vspace{.4cm}

\noindent In this section we obtain a linear kernel for  \textsc{Connected} \FSub,  \textsc{Connected} \lFSub, and \FSubM on the families of graphs excluding respectively a fixed graph, a fixed apex graph, and a fixed graph, as a minor.
\newline

For these three problems, the structure of a solution will be respectively a \emph{packing of subgraph models}, a \emph{packing of subgraph models}, and a \emph{packing of subgraph models with $\ell$-membership}. In order to define a packing of subgraph models, we need the definition of a \emph{subgraph model} of $F$ in $G$, which is basically an isomorphism from a graph $F$ to a subgraph of $G$.

\begin{definition} \label{defi: subG model}
A \emph{subgraph model} of a graph $F$ in a graph $G$ is a mapping $\Phi$, that assigns
to every vertex $v \in V(F)$ a vertex $\Phi(v) \in v(G)$, such that
\begin{itemize}
       \item[$\bullet$] the vertices $\Phi(v)$ for $v \in V(F)$ are distinct; and
       \item[$\bullet$] if $\{u,v\} \in E(F)$, then $\{\Phi(u),\Phi(v)\} \in E(G)$.
\end{itemize}

We denote by $\Phi(F)$ the subgraph of $G$ with vertex set $\{ \Phi(v): v \in V(F) \}$ and edge set $\{ \{\Phi(u),\Phi(v)\}: \{u,v \} \in E(F) \}$, which is obviously isomorphic to $F$.
\end{definition}

\begin{definition} \label{defi: packing models}
Let $\mc{F}$ be a set  of subgraphs and let $G$ be a graph.
A \emph{packing of subgraph models} $\mc{S}$ is a set of vertex-disjoint subgraph models, that is, the graphs $\Phi(F)$ for $\Phi \in \mc{S}, F \in  \mc{F}$ are vertex-disjoint.
A \emph{packing of subgraph models with $\ell$-membership} $\mc{S}$ is a set of subgraph models such that every vertex $v \in V(G)$ is the image of at most $\ell$ mappings $\Phi \in \mc{S}$.
\end{definition}

\subsection{A protrusion decomposition for an instance of \FSub}\label{ssec: Fsub prot decompo}

In order to find a linear protrusion decomposition, we first need a preprocessing reduction rule. This rule, which has also been used in previous work~\cite{FLST10,BFL+09}, enables us to establish a relation between instances of \FSub (and its variants) and \textsc{$d$-Dominating Set}. Then we will be able to apply Theorem \ref{theo: tw mod} on \textsc{$d$-Dominating Set} to find a linear treewidth-modulator that allows to construct the decomposition.


\begin{rgl} \label{rgl: prelim}
Let $v$ be a vertex of $G$ that does not belong to any subgraph of $G$ isomorphic to a graph in $\cal F$. Then remove $v$ from $G$.
\end{rgl}

Note that Rule~\ref{rgl: prelim}  can be applied in time $O(n^r)$, where $n$ is the size of $G$ and $r$ is the maximum size of a graph in $\cal F$. We call a graph\emph{ reduced} under Rule~\ref{rgl: prelim} if the rule cannot be applied anymore on $G$.

The next proposition states a relation between an instance of \lFSub and \textsc{$d$-Dominating Set}. The relation with the two other problems are straightforward, as explained below.
\begin{proposition} \label{prop: }
Let $G$ be a graph  reduced under Rule \ref{rgl: prelim}. If $(G,k)$ is a \NO-instance of \textsc{Connected} \lFSub, then $(G,k)$ is a \YES-instance of \textsc{$(2d+\ell)$-Dominating Set}, where $d$ is the largest diameter of the graphs in $\cal F$.
\end{proposition}

\begin{proof}
Let $(G,k)$ be a \NO-instance of \lFSub and let $d$ be the largest diameter of a graph in $\cal F$. Let us choose any vertex $v \in V(G)$ and remove $N_{d+\ell}(v)$ from $G$. We repeat this operation until there is no subgraph model of  $\cal F$ in $G$. We call $D$ the set of removed vertices. As $(G,k)$ is a \NO-instance of \FSub, $|D| \leq k$ and as  $G$ is reduced under Rule \ref{rgl: prelim} all vertices in $V(G) \setminus N_{d+\ell}(D)$ belong to a (connected) subgraph model (which intersects $N_{d+\ell}(D)$), hence all vertices in $V(G) \setminus N_{d+\ell}(D)$ are at distance at most $2d+\ell$ from $D$. Therefore $(G,k)$ is a \YES-instance of \textsc{$(2d+\ell)$-Dominating Set}.
\end{proof}

Note that if $(G,k)$ is a \NO-instance of \FSubM then it is a \NO-instance of \FSub (that is, of \textsc{$1$-$\mc{F}$-Subgraph-Packing}) and then it is a \NO-instance of \lFSub for every integer $\ell \geq 1$. According to Proposition \ref{prop: }, it follows that $(G,k)$ is a \YES-instance of \textsc{$(2d+1)$-Dominating Set}.
\newline
We now apply Theorem \ref{theo: tw mod} in order to find a treewidth-modulator for a \YES-instance of \textsc{$(2d+1)$-Dominating Set}. We now use the following corollary of Theorem~\ref{theo: tw mod}.

\begin{corollary} \label{coro: tw modul Dom}
Let $\mc{F}$ be a finite set of connected graphs, let $H$ be an $h$-vertex apex graph, and let \G be the class of $H$-minor-free graphs.
If $(G,k) \in \textsc{$d$-DS}_\G$, then there exists a set $X\subseteq V(G)$
  such that $|X|= k$                                          
  and       $\tw(G-X)= O (d \sqrt{d} \cdot \tau_H^3 \cdot f_c(h)^3 )$.           
Moreover, given an instance  $(G,k)$ with $|V(G)|=n$, there is an algorithm running in time $O(n^{3})$ that either finds such a set $X$ or correctly reports that $(G,k) \notin \textsc{$d$-DS}_\G$.
\end{corollary}

We are now able to construct a linear protrusion decomposition.

\begin{lemma} \label{lem: FSub prot decompo}
Let $\mc{F}$ be a finite set of connected graphs, let $H$ be an $h$-vertex apex graph, and let \G be the class of $H$-minor-free graphs.
Let $(G,k)$ be an instance of \textsc{Connected}-\lFSub (or of \textsc{Connected}-\FSub, or of \textsc{Connected}-\FSubM). If $(G,k)  \notin \textsc{c}\mc{F}\textsc{SP}_\G$, then we can construct in polynomial time a linear protrusion decomposition of $G$.
\end{lemma}

\begin{proof}
Given an instance $(G,k)$ of $\textsc{c}\mc{F}\textsc{SP}_\G$, we run the algorithm given by Corollary \ref{coro: tw modul Dom} for the \textsc{Connected $(2d+\ell)$-Dominating Set} problem, where $d$ is the largest diameter of the graphs in $\cal F$.
If the algorithm is not able to find a treewidth-modulator $X$ of size $|X|= k$, then by Proposition \ref{prop: } we can conclude that $(G,k) \in \textsc{c}\ell\mc{F}\textsc{SP}_\G$ (resp. $(G,k) \in \textsc{c}\mc{F}\textsc{SP}_\G$ and $(G,k) \in \textsc{c}\mc{F}\textsc{SP$\ell$M}_\G$).
Otherwise, we use the set $X$ as input to the algorithm given by Theorem \ref{theo:prot dec}, which outputs in linear time an
$((\alpha_{H} \cdot t)\cdot k, 2t + h)$-protrusion decomposition of $G$, where
\begin{itemize}
\item[$\bullet$] $t =O ( (2d+\ell)^{3/2} \cdot \tau_H^3 \cdot f_c(h)^3 )$ is provided by Corollary \ref{coro: tw modul Dom}; and
\item[$\bullet$] $\alpha_H = O( h^2 2 ^{ O(h \log h) })$ is  the constant provided by Theorem \ref{theo:prot dec}.
\end{itemize}

This is an $\left(h^2 2 ^{ O(h \log h) } \cdot  (2d+\ell)^{3/2} \cdot \tau_H^3 \cdot f_c(h)^3 \cdot k~,~ O ( (2d+\ell)^{3/2} \cdot \tau_H^3 \cdot f_m(h)^3 )\right)$-pro\-trusion decomposition of $G$.
\end{proof}

\subsection{An encoder for \FSub} \label{ssec: FSub encod}

Our encoder \E for \FSub uses a simplified version of rooted packings.

\begin{definition}
Let $F$ be a connected graph and let $G$ be a boundaried graph with boundary $B$.
A \emph{rooted set} of $B$ is a quadruple $(\A,S_F^*,S_F,\psi)$, where
\begin{itemize}
\item[$\bullet$] $S_F \subseteq S_F^*$ are both subsets of $V(F)$;
\item[$\bullet$] $\A$ is a non-empty subset of $B$; and
\item[$\bullet$]  $\psi: \A \to S_F$ is a bijective mapping assigning vertices of $S_F$ to the vertices in $\A$.
\end{itemize}

We also define a \emph{potential subgraph model} of $F$ in $G$ \emph{matching} with $(\A,S_F^*,S_F,\psi)$ as a partial mapping $\Phi$, that assigns
to every vertex $v \in S_F$ a vertex $\Phi(v) \in \A$ such that $\psi(\Phi(v))=v$, and
to every vertex $v \in S_F^*$ a vertex $\Phi(v) \in V( G)$ such that for all $u,v \in S_F^*$ if $\{u,v\} \in E(F)$ then, $\{\Phi(u), \Phi(v)\} \in E(G)$. Moreover, for every $v \in S_F^* \setminus S_F$, it holds that $\Phi(v) \in V(G) \setminus B$.
\end{definition}

Intuitively, the rooted set is a simplification of the rooted packing defined in Section \ref{sec: FPack}. The collection $\cal A$ of subsets of $B$ is replaced with a subset $A$ of $B$ (since now the image of a vertex $v \in V(F)$ is a vertex of $G$). The sets $ S_F^*,S_F$ still describe the subgraph of $F$ which is realized in $G$ and its vertices that lie in $B$. The function $\psi$ plays the same role as in rooted packings: it can be viewed as the inverse of the potential subgraph model $\Phi$ restricted to $B$. Note that we do not need the function $\chi$ anymore because the edges cannot appear later (because now the image of a vertex $v \in V(F)$ is a vertex, and we are dealing with a tree decomposition).

The number of distinct rooted sets at a separator $B$ is upper-bounded by $f(t,F):= 2^t \cdot r^t \cdot 2^{2r}$, where $t \geq |B|$ and $r= |V(F)|$.


Here, we only describe the encoder for \FSub. Similarly to Section~\ref{sec: rFPack}, the encoder for \lFSub is obtained by a combination of the encoder for \FSub and the one for $\ell$-\textsc{Scattered Set}. As in Section~\ref{sec: FMemb}, the encoder for \FSubM is obtained by allowing intersections in the rooted set.

\smallskip
\noindent\textbf{The encodings generator \C.}
Let $G \in \B$ with boundary $\partial (G)$ labeled with $\Lambda(G)$. The function \C maps $\Lambda(G)$ to a set  $\C(\Lambda(G))$ of encodings.
Each $R\in \C(\Lambda(G))$ is a set of at most $|\Lambda(G)|$ rooted sets $\{(A_i,S_{F_i}^*,S_{F_i} ,\psi_i) : F_i \in\mc{F} \}$, where each such rooted set encodes a potential subgraph model of $F_i \in \mc{F}$ (multiple subgraphs models of the same graph are allowed), and where the sets $A_i$ are pairwise disjoint.

\smallskip
\noindent\textbf{The language \L.}
For a packing of subgraph models $\mc{S}$, we say that $(G,\mc{S},R)$ belongs to the language $\L$ (or that $\mc{S}$ is a packing of models \emph{satisfying} $R$) if there is a packing of vertex-disjoint potential subgraph models matching with the rooted sets of $R$ in $G \setminus \bigcup_{\Phi \in \mc{S}} \Phi(F)$.

Note that we allow the entirely realized subgraph models of $\mc{S}$ to intersect $\partial (G) $ arbitrarily, but they must not intersect potential subgraph models imposed by $R$.
\newline

As in the previous sections, we need to use the relevant function \fbar. To this aim, we need to remark that, given a separator $B$ and a subgraph $G_B$, a (partial) solution naturally induces an encoding $R_B \in \C(\Lambda(G_B))$ where the rooted sets correspond to the intersection of models with $B$.

Formally, let $G$ be a $t$-boundaried graph with boundary $A$ and let $\mc{S}$ be a partial solution satisfying some $R_A \in \C(\Lambda(G))$. Let also $\mc{P}$ be the set of potential subgraph models matching with the rooted set in $R_A$. Given a separator $B$ in $G$, we define the induced encoding $R_B = \{(A_i,S_{F_i}^*,S_{F_i} ,\psi_i) : \Phi_i \in \mc{S}\cup\mc{P}\} \in \C(\Lambda(G_B))$ such that
for each (potential) subgraph model $\Phi_i \in \mc{S} \cup \mc{P}$ of $F_i \in \mc{F}$ intersecting $B$,
\begin{itemize}
\item[$\bullet$] $A_i$ contains vertices of $\Phi_i(F_i)$ in $B$;
\item[$\bullet$] $\psi_i$ maps each vertex of $A_i$ to its corresponding vertex in $F_i$; and
\item[$\bullet$] $S_{F_i}^*$ and $S_{F_i}$ correspond to the vertices of $F_i$ whose images by $\Phi$ belong to $G_B$ and $B$, respectively.
\end{itemize}
Clearly, the set of models of $\mc{S}$ entirely realized in $G_B$ is a partial solution satisfying $R_B$.

The definition of an irrelevant encoding is the same as in Section~\ref{sec:applications}.

\smallskip
\noindent\textbf{The function \fbar.}
Let $G \in \B$ with boundary $A$. We define the function
$\fbar$ as 
\begin{equation*} \label{eq: relevant f}
\fbar(G,R_A) =\
\left\{\begin{array}{lll}
  & -\infty,  & \text{if } \f(G,R_A) + t < \\
  &           & ~~~\max \{\f(G,R): R \in \C(\Lambda(G)) \} \\
  &           & ~~~\text{or if } R_A \text{ is irrelevant for } \fbar. \\
  & \f(G,R),  & \mbox{otherwise}.\\
\end{array}\right.
\end{equation*}

In the above equation, \f is the natural maximization function, that is $\f(G,R)$ is the maximal number of (entirely realized) subgraph models in $G$ which do not intersect potential subgraph models imposed by $R$. Formally,
\begin{equation*}\label{eq:fEmin}
\f(G,R) \ = \ \max \{k  :  \exists \mc{S}, |\mc{S}| \geq k, (G,\mc{S},R) \in \L\}.
\end{equation*}

\smallskip
\noindent\textbf{The size of \E.}
Recall that $f(t,F):= 2^t \cdot r^t \cdot 2^{2r}$ is the number of rooted sets for a subgraph $F$ of size $r$ on a boundary of size $t$. Our encoder contains at most $t$ vertex-disjoint rooted sets, for  subgraphs of size at most  $r := \max_{F \in \mc{F}} |V(F)|$ and such that the sum of their boundary size is at most $t$. Hence we can bound the size of the encoder as
\begin{equation*}
s_{\E}(t) \ \leq \ \big(\sum_{j\in J} 2^{j} \cdot r^j \cdot 2^{2r}\big)
            \ \leq \  t \cdot 2^{t} \cdot r^t \cdot 2^{2r} .
\end{equation*}

Note that the encoder for \lFSub generates couples of encodings for \FSub and $\ell$-\Scat, and therefore the size of the encoder can be bounded as
\begin{equation*}
s_{\El}(t) \ \leq \ s_{\E}(t) \cdot (\ell+2) ^ {t^2}.
\end{equation*}

Finally, note that the encoder for \FSubM contains at most $\ell t$ rooted sets on a boundary of size $t$, and thus the size of the encoder can be bounded as
\begin{equation*}
s_{\Em}(t) 
            \ \leq \  \ell t \cdot 2^{t} \cdot r^t \cdot 2^{2r}.
\end{equation*}

Similarly to Fact \ref{fait: rp}, the following fact claims that rooted sets allow us to glue and unglue boundaried graphs, preserving the existence of subgraphs. We omit the proof as it is very similar to the one of  Fact \ref{fait: rp}.

\begin{fact} \label{fait: rs}
Let $G \in \B$ with boundary $A$, let $\Phi$ be a subgraph model (resp. a potential subgraph model matching with a rooted set defined on $A$) of a graph $F$ in $G$, let $B$ be a separator of $G$, and let $G_B \in \B$ be as in Definition \ref{defi:DP-friend}. Let $(A, S_{F}^*,S_{F} ,\psi)$ be the rooted set induced by $\Phi$ (as defined above). Let $G_B' \in \B$ with boundary $B$ and let $G'$ be the graph obtained by replacing $G_B$ with $G_B'$.
If $G_B'$ has a potential subgraph model $\Phi_B'$ matching with $(A, S_{F}^*,S_{F} ,\psi)$, then $G'$ has a subgraph model (resp. a potential subgraph model) of $F$.
\end{fact}

We now have to prove that the encoders \E, \El, \Em are confined and DP-friendly. The proofs are very similar to the proof of Lemma \ref{lem: Pack DPfriend}; the proofs for \El and \Em have to be adapted following Sections \ref{sec: rFPack} and \ref{sec: FMemb}, respectively. This seems natural as the encoder \E is defined with rooted sets, which are simplifications of rooted packings.

\begin{lemma} \label{lem: Sub DPfriend}
The encoders \E, \El, and \Em are $g$-confined for $g(t)=t$, $g(t)=2t$, and $g(t)=t$, respectively. They are respectively a $\textsc{c}\mc{F}\textsc{SP}$-encoder, a $\textsc{c}\ell\mc{F}\textsc{SP}$-encoder, and a $\textsc{c}\mc{F}\textsc{SP}\ell\sc{M}$-encoder. Furthermore, if \G is an arbitrary class of graphs, then the equivalence relations \eqsG\E, \eqsG\El, and \eqsG\Em are DP-friendly.
\end{lemma}

\subsection{A linear kernel for \FSub} \label{ssec: FSub kernel}

We are now ready to provide a linear kernel for \textsc{Connected} \FSub, \textsc{Connected} \lFSub, and \textsc{Connected} ${\cal F}${\sc -Subgraph-Packing with} 
{\sc $\ell$-membership}.

%

\begin{theorem} \label{theo: Kernel FSub}

Let $\mc{F}$ be a finite family of connected graphs with diameter at most $d$, let $H$ be an $h$-vertex graph, and let \G be the class of $H$-minor-free graphs. Then \textsc{Connected} \FSub, \textsc{Connected} \lFSub, and \textsc{Connected $\mc{F}$-Subgraph Packing with $\ell$-Membership} admit constructive kernels of size $O(k)$, where the constant hidden in the ``$O$'' notation depends on $h$, $d$, and $\ell$.
\end{theorem}

The proof is similar to the ones in the previous sections. Using the protrusion decomposition given by Lemma \ref{lem: FSub prot decompo} and the encoders described in Section \ref{ssec: FSub encod}, we have all the material to apply Corollary \ref{coro:main}. The size of the kernel differs from the previous sections due to the size of the encoders and due to the bound on the treewidth of protrusions given by Lemma \ref{lem: FSub prot decompo}.

\smallskip

To conclude, we would like to mention that Romero and L{\'{o}}pez{-}Ortiz~\cite{RomeroL14-WALCOM} introduced another problem allowing intersection of subgraph models, called \textsc{$\mc{F}$-(Subgraph)-Packing with $\ell$-Overlap}. In this problem, also studied in~\cite{FernauLR15,RomeroL14-CSR}, a subgraph model can intersect any number of other models, but they are allowed to pairwise intersect on at most $\ell$ vertices. It is easier to perform dynamic programming on the membership version than on the overlap version, since the intersection constraint is local for the first one (just on vertices) but global for the second one (on pairs of models). However, we think that it is possible to define an encoder (with all the required properties) for \textsc{$\mc{F}$-(Subgraph)-Packing with $\ell$-Overlaps} using rooted sets and vectors of integers counting the overlaps (similarly to \lFSub). This would imply the existence of a linear kernel for the \textsc{$\mc{F}$-(Subgraph)-Packing with $\ell$-Overlap} problem on sparse graphs. We leave it for further research.

\section{Conclusions and further research}
\label{sec:conclusions}

In this article we generalized the framework introduced in~\cite{KviaDP} to deal with packing-certifiable problems. Our main result can be seen as a  {\sl meta-theorem}, in the sense that as far a particular problem satisfies the generic conditions stated in Corollary~\ref{coro:main}, an explicit linear kernel on the corresponding graph class follows. Nevertheless, in order to verify these generic conditions and, in particular, to verify that the equivalence relation associated with an encoder is DP-friendly, the proofs are usually quite technical
and one first needs to get familiar with several definitions. We think that it may be possible to simplify the general methodology, thus improving its applicability.

Concerning the explicit bounds derived from our results, one natural direction is to reduce them as much as possible. These bounds depend on a number of intermediate results that we use along the way and improving any of them would result in an improvement on the overall kernel sizes. It is worth insisting here that some of the bounds involve the (currently) {\sl non-explicit} function $f_c$ defined in Proposition~\ref{prop:tw-contraction}, which depends exclusively on the considered graph class (and not on each particular problem).  In order to find explicit bounds for this function $f_c$, we leave as future work using the linear-time deterministic protrusion replacer recently introduced by Fomin \emph{et al}.~\cite{FLM+15}, partially inspired from~\cite{KviaDP}.


\vspace{.3cm}
\noindent\textbf{Acknowledgement.} We would like to thank Archontia C. Giannopoulou  for insightful discussions about the Erd\H{o}s-P\'{o}sa property for scattered planar minors.



\bibliographystyle{abbrv}
\bibliography{linearkernels}

\newpage

\begin{appendix}
\section{Deferred proofs in Section~\ref{sec:generic}}
\label{ap:framework}

\subsection{Proof of Lemma~\ref{lem:nb class}}
\label{ap:lem:nb class}

Let us first show that the equivalence relation $\eqs\E$ has finite index. Let $I \subseteq \{1,\ldots,t\}$. Since we assume that  \E is $g$-confined, we have that for any $G \in \B$ with $\Lambda(G)=I$, the function $\f(G,\ \cdot\ )$ can take at most $g(t)+2$ distinct values ($g(t)+1$ finite values and possibly the value $-\infty$).
Therefore, it follows that the number of equivalence classes of \eqs\E containing all graphs $G \in \B$ with $\Lambda(G)=I$ is at most  ${(g(t)+2)^{|\C(I)|}}$. As the number of subsets of  $\{1,\ldots,t\}$ is $2^t$,
we deduce that the overall number of equivalence classes of \eqs\E is at most ${(g(t)+2)^{s_{\E}(t)}} \cdot 2^t$. Finally, since the equivalence relation \eqsG\E is the Cartesian product of the equivalence relations \eqs\E and \eq\G, the result follows from the fact that \G can be expressed in MSO logic.

\subsection{Proof of Fact~\ref{fait:equiv}}
\label{ap:fait:equiv}

Let $G = G^- \oplus G_B$  and let $G' = G^- \oplus G_B'$.
Assume that $G \eqs\E G'$. In order to deduce that $G \eqsG\E G'$, it suffices to prove that $G \eq\G G'$. Let $H\in \B$. We need to show that $G \oplus H \in \G$ if and only if $G' \oplus H \in \G$.
We have that $G \oplus H = (G_B \oplus G^-) \oplus H = G_B \oplus (G^- \oplus H)$, and similarly for $G'$. Since $G_B \eq\G G_B'$, it follows that
$G \oplus H = G_B \oplus (G^- \oplus H) \in \G$ if and only if $G_B \oplus (G^- \oplus H) = G \oplus H \in \G$.

\subsection{Proof of Lemma~\ref{lem:refine eq}}
\label{ap:lem:refine eq}

Let $\E = (\C,\L,\f)$ be a $\Pi$-encoder and let $G_1,G_2 \in \B$ such that $G_1 \eq\E G_2$.
We need to prove that for any $H\in \B$ and any integer $k$, $(G_1 \oplus H, k ) \in \Pi$ if and only if $(G_2 \oplus H, k + \D\E(G_1,G_2)) \in \Pi$.

Suppose that  $(G_1 \oplus H, k ) \in \Pi$ (by symmetry the same arguments apply starting with $G_2$). Since $G_1 \oplus H$ is a $0$-boundaried graph and \E is a $\Pi$-encoder, we have that
\begin{equation}\label{eq:refine eq}
\f (G_1 \oplus H,R_{\emptyset}) = f^{\Pi}(G_1 \oplus H) \geq k.
\end{equation}
As \eqsG\E is DP-friendly and $G_1 \eqsG\E G_2$, it follows that  $(G_1 \oplus H) \eqsG\E (G_2 \oplus H)$ and that
 $\D\E(G_1 \oplus H,G_2 \oplus H) = \D\E(G_1,G_2)$. Since $G_2 \oplus H$ is also a $0$-boundaried graph, the latter property and Equation \eqref{eq:refine eq} imply that
 \begin{equation}\label{eq:refine eq 2}
 \f(G_2 \oplus H,R_\emptyset) = \f(G_1 \oplus H,R_\emptyset) + \D\E(G_1,G_2) \geq k + \D\E(G_1,G_2).
 \end{equation}
Since \E is a $\Pi$-encoder, $f^\Pi(G_2 \oplus H) = \f(G_2 \oplus H,R_\emptyset)$, and from Equation \eqref{eq:refine eq 2} it follows that $(G_2 \oplus H, k + \D\E(G_1,G_2)) \in \Pi$.

\subsection{Proof of Lemma~\ref{lem:progres size}}
\label{ap:lem:progres size}

Let $\mathfrak{C}$ be an arbitrary equivalence class of \eqG\E, and let $G_1,G_2 \in \mathfrak{C}$. 
Let us first argue that $\mathfrak{C}$ contains some progressive representative. Since $\D\E(G_1,G_2) = \f(G_1,R) -\f(G_2,R)$ for every encoding $R$ such that $\f(G_1,R),\f(G_2,R) \neq - \infty$, $G \in \mathfrak{C}$ is progressive if $\f(G,R)$ is minimal in $\f(\mathfrak{C},R)= \{f(G,R) : G \in \mathfrak{C} \}$ for every encoding $R$ (including those for which the value is $- \infty$). Since $\f(\mathfrak{C},R)$ is a subset of $\mathbb{N} \cup \{ - \infty\} $, it necessarily has a minimal element, hence there is a progressive representative in $\mathfrak{C}$ (in other words, the order defined by $G_1
\preccurlyeq G_2$ if $\D\E(G_1,G_2)\leq 0$ is well-founded).

Now let $G \in \mc{G}$ be a progressive representative of $\mathfrak{C}$ with minimum number of vertices. We claim that $G$ has size at most $2^{r(\mathcal{E},g,t,\mathcal{G})+1} \cdot t$ (we would like to stress that at this stage we only need to care about the {\sl existence} of such representative $G$, and not about how to {\sl compute} it). Let $(T,\mc{X})$ be a boundaried nice tree decomposition of $G$ of width at most $t-1$ such that $\partial(G)$ is contained in the root-bag (such a nice tree decomposition exists by \cite{Klo94}).

We first claim that for any node $x$ of $T$, the graph $G_x$ is a progressive representative of its equivalence class with respect to \eqG\E, namely $\mathfrak{C'}$. Indeed, assume for contradiction that $G_x$ is not progressive, and therefore we know that there exists $G_x' \in \mathfrak{C'}$ such that $\D\E(G_x',G_x) < 0$. Let $G'$ be the graph obtained from $G$ by replacing $G_x$ with $G_x'$. Since \eqsG\E is DP-friendly, it follows that  $G \eqG\E G'$ and that $\D\E(G',G) = \D\E(G_x',G_x) < 0$, contradicting the fact that $G$ is a progressive representative of the equivalence class $\mathfrak{C}$.

We now claim that for any two nodes $x,y \in V(T)$ lying on a path from the root to a leaf of $T$, it holds that $G_x \neqG\E G_y$. Indeed, assume for contradiction that there are two nodes $x,y \in V(T)$ lying on a path from the root to a leaf of $T$ such that $G_x \eqG\E G_y$. Let $\mathfrak{C'}$ be the equivalence class of $G_x$ and $G_y$ with respect to \eqG\E. By the previous claim, it follows that both $G_x$ and $G_y$ are progressive representatives of $\mathfrak{C'}$, and therefore it holds that $\D\E(G_y,G_x) = 0$. Suppose without loss of generality 
that $G_y \subsetneq G_x$ (that is, $G_y$ is a strict subgraph of $G_x$), and let $G'$ be the graph obtained from $G$  by replacing $G_x$ with $G_y$. Again, since \eqsG\E is DP-friendly, it follows that  $G \eqG\E G'$ and that $\D\E(G',G) = \D\E(G_y,G_x) = 0$. Therefore, $G'$ is a progressive representative of $\mathfrak{C}$ with $|V(G')| < |V(G)|$, contradicting the minimality of $|V(G)|$.

Finally, since for any two nodes $x,y \in V(T)$ lying on a path from the root to a leaf of $T$ we have that $G_x \neqG\E G_y$, it follows that the height of $T$ is at most the number of equivalence classes of \eqG\E, which is at most $r( \E,g,t,\G)$ by Lemma \ref{lem:nb class}. Since $T$ is a binary tree, we have that $|V(T)| \leq 2^{r(\mathcal{E},g,t,\mathcal{G})+1} - 1$. Finally, since $|V(G)| \leq |V(T)| \cdot t$, it follows that $|V(G)| \leq 2^{r(\mathcal{E},g,t,\mathcal{G})+1} \cdot t$, as we wanted to prove.

\subsection{Proof of Lemma~\ref{lem:comput repres}}
\label{ap:lem:comput repres}

Let $\E = (\C,\L,\f)$ be the given encoder. We start by generating a repository $\mathfrak{R}$ containing all the graphs in \F with at most $b+1$ vertices. Such a set of graphs, as well as a boundaried nice tree decomposition of width at most $t-1$ of each of them, can be clearly generated in time depending only on $b$ and $t$. By assumption, the size of a smallest progressive representative of any equivalence class of \eqG\E is at most $b$, so  $\mathfrak{R}$ contains a progressive representative of any equivalence class of \eqG\E with at most $b$ vertices.
We now partition the graphs in $\mathfrak{R}$ into equivalence classes of \eqG\E as follows: for each graph $G \in \mathfrak{R}$ and each encoding $R \in \C(\Lambda(G))$, as \L and \f are computable, we can compute the value $\f(G,R)$ in time depending only on $\E,g,t,$ and $b$. Therefore, for any two graphs $G_1,G_2 \in \mathfrak{R}$, we can decide in time depending only on $\E,g,t,b$, and \G whether $G_1 \eqG\E G_2$, and if this is the case, we can compute the transposition constant $\D\E(G_1,G_2)$ within the same running time.

Given a $t$-protrusion $Y$ on $n$ vertices with boundary $\partial(Y)$, we first compute a boundaried nice tree decomposition $(T,\mc{X},r)$ of $Y$ in time $f(t) \cdot n$, by using the linear-time algorithm of Bodlaender~\cite{Bod96,Klo94}.
Such a $t$-protrusion $Y$ equipped with a tree decomposition can be naturally seen as a $t$-boundaried graph by assigning distinct labels from $\{1,\ldots,t\}$ to the vertices in the root-bag.
We can assume that $\Lambda(Y)=\{1,\ldots,t\}$. Note that the labels can be transferred to the vertices in all the bags of $(T,\mc{X},r)$, by performing a standard shifting procedure when a vertex is introduced or removed from the nice tree decomposition \cite{BFL+09}. Therefore, each node $x \in V(T)$ defines in a natural way a $t$-protrusion  $Y_x \subseteq Y$ with its associated boundaried nice tree decomposition, with all the boundary vertices contained in the root bag. Let us now proceed to the description of the replacement algorithm.

We process the bags of $(T,\mc{X})$ in a bottom-up way until we encounter the first node $x$ in $V(T)$ such that $|V(Y_x)|=b+1$ (note that as $(T,\mc{X})$ is a nice tree decomposition, when processing the bags in a bottom-up way, at most one new vertex is introduced at every step, and recall that $b \geq t$, hence such an $x$ exists). We compute the equivalence class $\mathfrak{C}$ of $Y_x$ according to \eqG\E; this corresponds to computing the set of encodings $\C(\Lambda(Y_x)) $ and the associated values of $\f(Y_x,\cdot)$ that, by definition of an encoder, can be calculated since $\f$ is a computable function.
 As $|V(Y_x)|=b+1$, the graph $Y_x$ is contained in the repository $\mathfrak{R}$, so in constant time we can find in  $\mathfrak{R}$ a progressive representative $Y_x'$ of $\mathfrak{C}$ with at most $b$ vertices and the corresponding transposition constant $\D\E(Y_x',Y_x) \leq 0$, (the inequality holds because $Y_x'$ is progressive). Let $Z$ be the graph obtained from $Y$ by replacing $Y_x$ with $Y_x'$, so we have that $|V(Y)| < |V(Z)|$ (note that this replacement operation directly yields a boundaried nice tree decomposition of width at most $t-1$ of $Z$). Since \eqsG\E is DP-friendly, it follows that  $Y \eqG\E Z$ and that $\D\E(Z,Y) = \D\E(Y_x',Y_x) \leq 0$.

We recursively apply this replacement procedure on the resulting graph until we eventually obtain a $t$-protrusion $Y'$ with at most $b$ vertices such that $Y \eqG\E Y'$. The corresponding transposition constant $\D\E(Y',Y)$ can be easily computed by summing up all the transposition constants given by each of the performed replacements. Since each of these replacements introduces a progressive representative, we have that $\D\E(Y',Y) \leq 0$. As we can assume that the total number of nodes in a nice tree decomposition of $Y$ is $O(n)$ \cite[Lemma 13.1.2]{Klo94}, the overall running time of the algorithm is $O(n)$ (the constant hidden in the ``$O$'' notation depends indeed exclusively on $\E,g,b,\G$, and $t$).





\subsection{Proof of Theorem~\ref{theo:main}}
By Lemma \ref{lem:nb class}, the number of equivalence classes of the equivalence relation \eqG\E is finite and by Lemma \ref{lem:progres size} the size of a smallest progressive representative of any equivalence class of \eqG\E is at most $b(\E,g,t,\G)$.
Therefore, we can apply Lemma \ref{lem:comput repres} and deduce that, in time $O(|Y|)$, we can find a $t$-protrusion $Y'$ of size at most $b(\E,g,t,\G)$ such that $Y \eqG\E Y'$ and the corresponding transposition constant $\D\E(Y',Y)$ with $\D\E(Y',Y) \leq 0$.
Since \E is a  $\Pi$-encoder and \eqsG\E is DP-friendly, it follows from Lemma \ref{lem:refine eq} that $Y \equiv_{\Pi} Y'$ and that $\D\Pi(Y',Y) = \D\E(Y',Y) \leq 0$.
Therefore, if we set $k' := k +\D\Pi(Y',Y)$, it follows that $(G,k)$ and $((G - (Y-\partial(Y)))\oplus Y',k')$ are indeed equivalent instances of $\Pi$ with $k' \leq k$ and $|Y'| \leq b(\E,g,t,\G)$.

\subsection{Proof of Corollary~\ref{coro:main}}
For $1 \leq i \leq \ell$, where $\ell$ is the number of protrusions in the decomposition, we apply the polynomial-time algorithm given by Theorem \ref{theo:main} to replace each $t$-protrusion $Y_i$ with a graph 
$Y_i'$ of size at most $b(\E,g,t,\G)$ and to update the parameter accordingly. In this way we obtain an equivalent instance $(G',k')$ such that $G' \in  \G$, $k' \leq k$ and $|V(G')| \leq |Y_0| + \ell \cdot b(\E,g,t,\G) \leq (1+b(\E,g,t,\G))\alpha \cdot k$ .

\end{appendix}

\end{document}